\documentclass[final,leqno,onefignum,onetabnum]{siamltex1213}

\usepackage{epsfig}
\usepackage{graphicx}

\usepackage{amssymb}
\usepackage{amsmath}

\newtheorem{remark}[theorem]{Remark}

\newcommand\cE{{\cal E}}
\newcommand\cH{{\cal H}}
\newcommand\cG{{\cal G}}
\newcommand\cF{{\cal F}}

\newcommand\cL{{\cal L}}

\newcommand\cK{{\cal K}}
\newcommand\cN{{\cal N}}

\newcommand\cX{{\cal X}}

\newcommand\cV{{\cal V}}

\newcommand\e{\epsilon}

\newcommand\ve{\varepsilon}
\newcommand\ov{\overline}

\def\bbr{{\mathbb R}}
\def\bbn{{\mathbb N}}

\def\text#1{\hbox{#1}}

\def\a{{\bf a}}
\def\B{{\bf B}}

\def\b{{\bf b}}
\def\q{{\bf q}}
\def\r{{\bf r}}
\def\p{{\bf p}}

\def\h{{\bf h}}

\def\g{{\bf g}}
\def\E{{\bf E}}
\def\P{{\bf P}}

\def\C{{\bf C}}

\def\G{{\bf G}}

\def\M{{\bf M}}

\def\U{{\bf U}}

\def\k{{\bf k}}

\newcommand{\wh}{\widehat}
\newcommand{\wt}{\widetilde}

\def\d{\mathrm{d}}
\def\build #1_#2{\mathrel{\mathop{\kern 0pt #1}\limits_{#2}}}
\newcommand{\zs}[1]{{\mathchoice{#1}{#1}{\lower.25ex\hbox{$\scriptstyle#1$}}
{\lower0.25ex\hbox{$\scriptscriptstyle#1$}}}}

\title{Sequential $\delta$-optimal consumption and investment for stochastic
volatility markets with unknown parameters
\thanks{The research is funded by the grant of the Government of Russian Federation 
No.14.A.12.31.0007
and  by the  Russian Science Fondation (research project No. 14-49-00079)
 and the second author is partially supported 
 by the International Laboratory of Statistics of Stochastic Processes and Quantitative Finance of Russian National Research Tomsk State 
 }
}

\author{Belkacem Berdjane\thanks{
D\'epartement de math\'ematiques, universit\'e du Qu\'ebec a Montr\'eal, Qu\'ebec,Canada. \newline
Email:\email{berdjane.belkacem@uqam.ca}. }
\and
Serguei Pergamenshchikov\thanks{
 Laboratoire de Math\'ematiques Raphael Salem,
 Avenue de l'Universit\'e, BP. 12,
  Universit\'e de Rouen,
   F76801, Saint Etienne du Rouvray, Cedex France
and
International Laboratory of Quantitative Finance, National Research University  Higher School of Economics,
   Moscow, Russia. \newline
Email:\email{Serge.Pergamenchtchikov@univ-rouen.fr} }}

\begin{document}
\maketitle

\begin{abstract}
We consider an optimal investment and
consumption problem for a Black-Scholes financial market with stochastic volatility and
unknown stock price appreciation rate. The volatility parameter is driven by an external economic
factor modeled as a diffusion process of Ornstein-Uhlenbeck type with unknown drift.
We use the dynamical programming approach and find an optimal financial strategy which
depends on the drift parameter.
To estimate the drift coefficient we observe the economic factor $Y$ in an interval $[0,T_0]$ for fixed $T_0>0$,
and use sequential estimation. We show that the consumption and investment strategy
calculated through this sequential procedure is $\delta$-optimal.
\end{abstract}

\begin{keywords}
Sequential analysis, Truncate sequential estimate, Black-Scholes market model, Stochastic volatility, Optimal Consumption and Investment, Hamilton-Jacobi-Bellman equation.
\end{keywords}

\begin{AMS}
62L12, 62L20, 91B28, 91G80, 93E20.
\end{AMS}

\pagestyle{myheadings}
\thispagestyle{plain}
\markboth{B. Berdjane and S. Pergamenshchikov}{Sequential $\delta$-optimal consumption and investment}


\section{Introduction}\label{sec:In}

We deal with the finite-time optimal consumption and investment problem in a Black-Scholes financial
market with stochastic volatility (see, e.g., \cite{FouquePapanicoloauSircar2000}).
We consider the same power utility function for both
 consumption and terminal wealth.
The volatility parameter in our situation depends on some economic factor, modeled as a diffusion process of Ornstein-Uhlenbeck type.
The classical approach to this problem goes back to Merton \cite{Merton1971}.

\vspace{1mm}
By applying results from the stochastic control, explicit solutions have been obtained for
financial markets with nonrandom coefficients (see, e.g. \cite{KaratzasShreve1998},
\cite{Korn1997}, \cite{SethiTaksarPresman1992}, \cite{PresmanSethi1991}   and references therein). Since then, the  consumption and investment problems
has been extended in many directions \cite{Rogers2013}.
One of the important generalizations considers financial models with stochastic volatility,
 since empirical studies of stock-price returns show that the estimated volatility
exhibits random characteristics (see e.g., \cite{Rubinstein1985} and
 \cite{JackwerthRubinstein1996}).

\vspace{1mm}
The pure investment problem for such models is considered in \cite{Zariphopoulou2001} and
\cite{Pham2002}. In these papers, authors use the dynamic programming
approach and show that the nonlinear HJB (Hamilton-Jacobi-Bellman) equation can be transformed
into a quasilinear PDE. The similar approach has been used in \cite{KraftSteffensen2006}
for optimal consumption-investment problems with the default risk
for financial markets with non random coefficients. Furthermore, in
\cite{FlemingHernandezHernandez2003}, by making use of the Girsanov measure
transformation the authors study a pure optimal consumption problem
for stochastic volatility markets. In \cite{CastanedaLeyvaHernandezHernandez2005} and
\cite{HernandezHernandezShied2006} the authors use  dual methods.

\vspace{1mm}
Usually, the classical existence and uniqueness theorem
for the HJB equation is shown by the linear PDE methods (see, for example, chapter VI.6
and appendix E in \cite{FlemingRishel1975}). In this paper we use the approach proposed
in \cite{DelongKluppelberg2008}
and used in \cite{BerdjanePergamen2012}. The difference between our work and these two
papers is that, in \cite{DelongKluppelberg2008}, authors consider a pure jump process as
the driven economic factor. The HJB equation in this case is an integro-differential equation
of the first order. In our case it is a highly non linear PDE of the second order.
In \cite{BerdjanePergamen2012} the same problem is considered where the market coefficients are
known, and depend on a diffusion process with bounded parameters. The result therein does not
allow the Gaussian Ornstein-Uhlenbeck process. Similarly to
\cite{DelongKluppelberg2008} and \cite{BerdjanePergamen2012} we study the HJB equation  through
the Feynman - Kac representation. We introduce a special metric space in which the Feynman - Kac
 mapping is a contraction. Taking this into account we show the fixed-point theorem for
this mapping and we show that the fixed-point solution
is the classical unique solution for the HJB equation in our case.

\vspace{1mm}

In the second part of our paper, we consider both the
stock price appreciation rate and the drift of the economic factor to be unknown.
To estimate the drift of a process of Ornstein-Uhlenbeck type
we require sequential analysis methods (see \cite{Novikov1971}
and \cite{LiptserShiryaev2000-II}, Sections 17.5-6).
The drift parameter will be estimated from the observations of the process $Y$,
in some interval $[0,T_0]$. It should be noted that in this case the usual
likelihood estimator for the drift parameter is a nonlinear function of observations
and it is not possible to calculate directly a non-asymptotic upper bound for
its accuracy. To overcome this difficulty we use the truncated sequential estimate from \cite{PergamenKonev92} which enables us a non-asymptotic upper bound for mean accuracy estimation.
 After that we deal with the optimal strategy in the interval $[T_\zs{0},T]$, under the estimated parameter. We show that the expected absolute deviation of the objective function for the given strategy is less
 than some known fixed level $\delta$ i.e. the strategy calculated through the sequential procedure is $\delta$-optimal. Moreover, in this paper we find the explicit form for this level. This allows
 to keep small the deviation of the objective function from the optimal one by controlling the initial endowment.

\vspace{1mm}

The paper is organized as follows: In Sections \ref{sec:Market-Model}-\ref{sec:HJB}
we introduce the market model, state the optimization problem and give the related HJB equation.
Section \ref{sec:usefull definitions} is set for definitions.
The solution of the optimal consumption and investment problem is given in Sections
\ref{sec:Solution-HJB}-\ref{sec:Optimal-Strategy}. In Section \ref{sec:delta-optimal strategy}
we consider unknown the drift parameter $\alpha$ for the economic factor $Y$ and use a truncated
sequential method to construct its estimate $\wh \alpha$.
We obtain an explicit upper for the deviation $\E \, |\wh \alpha- \alpha|$
for any fixed $T_0>0$.
 Moreover considering the optimal
consumption investment problem in the finite interval $[T_0,T]$, we show that the strategy calculated through this truncation procedure is  $\delta$-optimal.
Similar results are given in Section \ref{sec:mu-unknown} when, in addition of using $\wh \alpha$, we consider an estimate $\wh \mu$ of the unknown stock price appreciation rate. A numerical example is given in Section \ref{sec:simulation} and auxiliary results are reported into the appendix.

\section{Market model}\label{sec:Market-Model}

Let $(\Omega, \cF_\zs{T}, (\cF_\zs{t})_\zs{0\le t\le T},\P)$
be a standard and filtered probability space
with two standard independent $(\cF_\zs{t})_\zs{0\le t\le T}$ adapted
Wiener processes $(W_\zs{t})_\zs{0\le t\le T}$ and
$(\U_\zs{t})_\zs{0\le t\le T} $ taking their values in $\bbr$. Our financial market consists of one \emph{riskless} money market account  $(S_\zs{0}(t))_\zs{0\le t\le T}$ and
one \emph{risky stock}
$(S(t))_\zs{0\le t\le T}$ governed by the following equations:
\begin{equation}\label{eq:BS-model}
\left\{\begin{array}{ll}
\d S_\zs{0}(t)&=r\,S_\zs{0}(t)\, \d t\,,\\[5mm]
\d S(t)&=S(t)\mu \, \, \d t+S(t)\,
\,\sigma(Y_\zs{t})\,\d W_\zs{t}\,,
\end{array}\right.
\end{equation}
with $S_\zs{0}(0)=1$ and $S(0)=s >0$.
In this model
$r\,\in\bbr_\zs{+}$ is the {\em riskless bond interest rate},
$\mu$ is the {\em stock price appreciation rate} and
$\sigma(y)$ is {\em stock-volatility}.
For all  $y\in \bbr$
the coefficient $\sigma(y)\in\bbr_+$ is
a nonrandom continuous bounded function and satisfies
$$
\inf_{y \in \bbr} \sigma(y) =\sigma_1 > 0.
$$
We assume also that $\sigma(y)$ is differentiable and has bounded derivative. Moreover
 we assume that the stochastic factor $Y$ valued in $\bbr$
is of Ornstein-Uhlenbeck type. It has a dynamics governed by the
following stochastic differential equation:
\begin{equation}\label{sec:Mm.2}
\d Y_\zs{t}=\alpha Y_{t} \,\d t+\beta\d\U_{t}\,,
\end{equation}
where the initial value $Y_\zs{0}$ is a non random constant,
 $\alpha < 0$ and $\beta>0$ are fixed parameters.
We denote by $(Y^{t,y}_\zs{s})_\zs{s\ge t}$
the process $Y$ starts at $Y_t=y$, i.e.
$$
Y^{t,y}_s= y e^{\alpha (s-t)} + \int_t^s \beta e^{\alpha (s-v)} \, \d \U_v\,.
$$
We note,
that for the model \eqref{eq:BS-model} the risk premium is the
$\bbr \to \bbr\,$ function defined as
\begin{equation}\label{sec:Mm.3}
\theta(y)=\frac{\mu-r}{\sigma(y)}\,.
\end{equation}

For any $t\ge 0$ let  $\check{\phi}_\zs{t}\in\bbr$ denote  the amount of investment into bond and
$\wt{\phi}_\zs{t}\in\bbr$
the amount of investment into risky assets.
We recall  that a
{\em trading strategy} is an $\bbr^{2}$-valued $(\cF_\zs{t})_\zs{0\le t\le T}$-progressively measurable process
$(\check{\phi}_\zs{t},\wt{\phi}_\zs{t})_\zs{0\le t\le T}$ and that
$$
X_\zs{t}\,=\,\check{\phi}_\zs{t}\,S_\zs{0}(t)\,+
\wt{\phi}_\zs{t}\,S(t)\,,\quad t\ge 0\,,
$$
is called the {\em wealth process}.
Moreover, an $(\cF_\zs{t})_\zs{0\le t\le T}$-progressively measurable nonnegative
process $(\varsigma_\zs{t})_\zs{0\le t\le T}$ satisfying for the investment horizon $T>0$
$$
\int^T_0\,\varsigma_\zs{t}\,\d t\,<\,\infty \quad \mbox{a.s.}
$$
is called {\em consumption process}.
The trading strategy $(\check{\phi}_\zs{t},\wt{\phi}_\zs{t})_\zs{0\le t\le T}$ and the consumption process
$(\varsigma_\zs{t})_\zs{0\le t\le T}$ are called {\em self-financing}, if the wealth process
 satisfies the following equation
\begin{equation}\label{2.2-nn1}
X_\zs{t}\,=\,x\,+\,
\int^t_0\,\check{\phi}_\zs{u}\,\d S_\zs{0}(u)\,+
\,\int^t_0\,\wt{\phi}_\zs{u}\,\d S(u)\,-\,
\int^t_0\,\varsigma_\zs{u}\,\d u\,, \quad t\ge 0\,,
\end{equation}
where $x>0$ is the initial endowment.
Similarly to \cite{KluppelbergPergamenchtchikov2009}
we consider the fractional portfolio process $\varphi(t)=\wt{\phi}_\zs{t}S(t)/X_\zs{t}$, i.e. $\varphi(t)$ is the fraction of the wealth process $X_\zs{t}$ invested in the stock at the time $t$. The fraction for the consumption
is denoted by  $c_\zs{t}=\varsigma_\zs{t}/X_\zs{t}$. In this case  the wealth process
satisfies the following stochastic equation
\begin{equation}\label{eq:EDS-X}
 \d X_\zs{t}=X_\zs{t}(r \,+\pi_\zs{t}
\theta(Y_\zs{t})-c_\zs{t})
\, \d t+X_\zs{t}\pi_\zs{t}\, \d W_\zs{t}\,,
\end{equation}
where $\pi_\zs{t}= \sigma(\,Y_\zs{t})\, \varphi_\zs{t}$ and the initial endowment
$X_\zs{0}=x$.

 Now we describe the set of all admissible
strategies. A portfolio control (financial strategy)
$\vartheta=(\vartheta_\zs{t})_\zs{t\ge 0}=((\pi_\zs{t},c_\zs{t}))_\zs{t\ge 0}$
is  said to be {\em admissible} if it is $(\cF_\zs{t})_\zs{0\le t\le T}$ - progressively measurable with values in $\bbr\times [0,\infty)$, such that
\begin{equation}\label{sec:Mm.6}
\|\pi\|_\zs{T}
:=\int^{T}_\zs{0}
|\pi_\zs{t}|^2\, \d t
<\infty
\quad\mbox{and}\quad
\int^{T}_\zs{0}\,c_\zs{t}\, \d t
<\infty
\quad\mbox{a.s.}
\end{equation}
and equation \eqref{eq:EDS-X} has a unique strong a.s. positive continuous solution
$(X^{\vartheta}_\zs{t})_\zs{0\le t\le T}$ on $[0\,,\,T]$.
We denote the set of  {\em admissible portfolios controls} by $\cV$.

In this paper we consider an agent using the power
utility function $x^{\gamma}$ for $0<\gamma<1$.
The goal is to maximize the expected utilities from the consumption
 on the time interval $[T_0,T]$, for fixed $T_0$, and from the  terminal wealth at maturity $T$.
Then for any $x,y \in\bbr$, and $\vartheta\in\cV$ the
 value function is defined by
$$
J(T_0,x,y,\vartheta):=\E_\zs{T_0,x,y}\,\left(\int^T_\zs{T_0}
c_\zs{t}^{\gamma}\,(X^{\vartheta}_\zs{t})^{\gamma}\, \d t\,+\,(X^{\vartheta}_\zs{T})^{\gamma}\right)\,,
$$
were $\E_{T_0,x,y}\,$ is the conditional expectation $\E\,(\,.\,\vert X_{T_0}=x, Y_{T_0}=y)$.
Our goal is to maximize this function, i.e. to calculate
\begin{equation}\label{eq:Optimisation-Problem}
J(T_0,x,y,\vartheta^*)=\sup_\zs{\vartheta\in\cV}\,J(T_0,x,y,\vartheta)\,.
\end{equation}
For the sequel we will use the notations $J^*(T_0,x,y)$ or simply $J_{T_0}^*$ instead of $J(T_0,x,y,\vartheta^*)$.


\begin{remark}\label{Re.sec:Mm.1}
Note that the same problem as \eqref{eq:Optimisation-Problem} is solved in \cite{BerdjanePergamen2012}, but the economic factor $Y$ considered there is a general diffusion process with bounded coefficients. In the present paper $Y$ is an Ornstein-Uhlenbeck process, so the drift is not bounded, but we take advantage of the fact that $Y$ is Gaussian and not correlated to the market, which is not the case in \cite{BerdjanePergamen2012}.
\end{remark}

\section{Hamilton-Jacobi-Bellman equation}\label{sec:HJB}

Now we introduce the HJB equation for the problem
\eqref{eq:Optimisation-Problem}.
To this end, for any two times differentiable $[0,T]\times\bbr_\zs{+}\times\bbr\to\bbr$ function $f$ we  denote by
$D f(t,x,y)$ and $D^{2} f(t,x,y)$ the following vectors of the  partial derivatives i.e.
$$
 D f(t,x,y)=\left( \frac{\partial  }{\partial x} f(t,x,y)\,,\,\frac{\partial  }{\partial y} f(t,x,y)\right)'
$$
and
 $$
 D^{2} f(t,x,y)=\left( \frac{\partial^{2}  }{\partial x^{2}} f(t,x,y)\,,\,\frac{\partial^{2}  }{\partial y^{2}} f(t,x,y)\right)'
\,.
$$
Here the prime denotes the transposition.
Let now  $\q=(\q_1, \q_2)\in\bbr^{2}$, $\M=(\M_1, \M_2)\in\bbr^{2}$ and $\nu=(\nu_1, \nu_2)\in\bbr\times\bbr_\zs{+}$ be fixed parameters.
For these parameters we set
$$
H_\zs{0}(x,y,\q,\M,\nu):=
(r+\nu_\zs{1}\theta(y)-\nu_\zs{2})x\q_\zs{1}+
\alpha y\q_\zs{2}+\,
\frac{1}{2}\,M_\zs{1}\nu^{2}_\zs{1}
x^{2}
+
\frac{\beta^{2}}{2}\,M_\zs{2}
+
(\nu_\zs{2}x)^{\gamma}
\,.
$$
Now   we define the Hamilton function as
$$
H(x,y,\q,\M)\,:=\,\sup_{\nu\in \bbr\times\bbr_\zs{+}}\,H_{0}(x,y,\q,\M,\nu)\,.
$$
Note that, in this case for $x>0$, $\q_\zs{1}>0$ and $\M_\zs{1}<0$
\begin{align}
H(x,y,\q,\M)
&=
x\,r\,\q_\zs{1}+\alpha \,y\q_2
+\frac{1}{q_\zs{*}}
\left(\frac{\gamma}{\q_\zs{1}}\right)^{q_\zs{*}-1}\nonumber
+ \frac{|\theta(y)\q_\zs{1}|^{2}}{2|M_\zs{1}|}
+\frac{\beta^{2}}{2}\, M_\zs{2}\, \label{def:Hamilton-txy}\,,
\end{align}
where $q_\zs{*}=(1-\gamma)^{-1}$. The HJB equation is given by
\begin{equation}\label{eq:HJB-z}
 \left\{
\begin{array}{rl}
&\dfrac{\partial}{\partial t}z(t,x,y)+
H\left(x,y,Dz(t,x,y),D^2 z(t,x,y)\right)=0 \\[5mm]
&z(T,x,y)=x^{\gamma}\,.
\end{array}
\right.
\end{equation}
To study this equation we represent $z(t,x,y)$ as
\begin{equation}\label{sec:HJB.3}
z(t,x,y)=x^{\gamma}h(t,y)\,.
\end{equation}
It is easy to deduce that the function $h$ satisfies the following quasi-linear PDE:
\begin{equation}\label{eq:HJB-h}
\left\{
\begin{array}{rl}
\dfrac{\partial}{\partial t}h(t,y)&+
Q(y)\, h(t,y)
+
\alpha \,y \,\dfrac{\partial }{\partial y}h(t,y)+\dfrac{\beta^{2}}{2} \dfrac{\partial^2 }{\partial y^2}h(t,y)\\[5mm]
&+\dfrac{1}{q_\zs{*}}
\left(\dfrac{1}{h(t,y)}\right)^{q_\zs{*}-1}
=0\,;
\\[5mm]
h(T,y)&=1\,,
\end{array}
\right.
\end{equation}
where
\begin{equation}\label{def:q*_Q}
Q(\,y)=\gamma
\left( r\,+
\frac{\theta^{2}(y)}{2\left(1-\gamma\right)}\right)\,.
\end{equation}
Note that, using the conditions on $\sigma(y)$; the function $Q(y)$ is bounded differentiable and has bounded derivative. Therefore, we can set
\begin{equation}\label{def:Q*-Q_1*}
Q_\zs{*}=\sup_\zs{y\in \bbr}\,Q(\,y)
\quad\mbox{and}\quad
Q_\zs{1}^*=
\sup_\zs{y \in \bbr}\,
|\d Q(y)/\d y|\,.
\end{equation}

Our  goal is to study equation \eqref{eq:HJB-h}.
By making use of the probabilistic representation for
the linear PDE (the Feynman-Kac formula)
 we  show in Proposition~\ref{Pr.sec:Prl.5},  that the solution of this equation is
the fixed-point solution for a special mapping of the integral type
which will be introduced in the next section.

\section{Useful definitions}\label{sec:usefull definitions}

First, to study equation  \eqref{eq:HJB-h}
 we  introduce a special  functional space. Let $\cX$ be the set of uniformly continuous functions on $\cK:=[T_0,T]\times \bbr$ with values in $ [1,\infty)$ such that
\begin{equation}\label{def:norme_in_cX-0}
\|f\|_\infty=
\sup_\zs{(t,y)\in \cK}\,|f(t,y)|\,\le \r^{*}\,,
\end{equation}
where
$\r^{*}\,= \,(\wt T+1)\, e^{Q_\zs{*} \,\wt T}$
and $\wt T= T-T_\zs{0}$.
Now, we define a metric $\varrho_\zs{*}(.,.)$ in $\cX$ as follows: for any $f,g$ in $\cX$
\begin{equation}\label{def:norme_in_cX}
\varrho_\zs{*}(f,g)= \|f-g\|_\zs{*}
 =\sup_\zs{(t,y)\in \cK}
\,e^{-\varkappa(T-t)}\,|f(t,y)-g(t,y)|\,.
\end{equation}
Here $\varkappa=Q_\zs{*}+\zeta+1$ and  $\zeta$ is some positive parameter which will be specified later. We define now the process $\eta$ by its dynamics
\begin{equation}\label{def:EDS-eta}
\d \eta_s=\alpha \, \eta_s \, \d s+\,\beta  \,\d \wt \U_s
\quad \mbox{and}\quad \eta_0=Y_0\,.
\end{equation}
So, $(\eta_t)_\zs{t\ge 0}$ has the same distribution as $(Y_t)_\zs{t\ge 0}$. Here $(\wt \U_t)_\zs{t\ge 0}$ is a standard Brownian motion independent of $(\U_t)_\zs{t\ge 0}$.
Let's now define the $\cX \rightarrow \cX$ Feynman-Kac
mapping $\cL $:
\begin{equation}\label{def:L-theOperator}
\cL_\zs{f}(t,y)= \E\,\cG(t,T,y)
+\frac{1}{q_\zs{*}}
\int_{t}^{T}
\cH_\zs{f}(t,s,y)
\,\d s\,,
\end{equation}
where $\cG(t,s,y)= \exp \left( \int_{t}^{s} Q(\eta_{u}^{t,y})\, \d u\right)$ and
\begin{equation}\label{def:Hf(t,s,y)}
\cH_\zs{f}(t,s,y)=\E
\left(f(s,\eta^{t,y}_{s})\right)^{1-q_\zs{*}}
\cG(t,s,y)\,.
\end{equation}

\noindent and $(\eta_s^{t,y})_{t\le s\le T}$ is the process $\eta$ starting at $\eta_t=y$. To solve the  HJB equation we need to
find the fixed-point solution for the mapping $\cL$ in
$\cX$, i.e.
\begin{equation}\label{sec:Optimal_Strategy.7}
\cL_\zs{h}=h\,.
\end{equation}
To this end we construct
the following  iterated scheme. We set
$h_\zs{0}\equiv 1$
\begin{equation}\label{def:hn-sequence}
h_\zs{n}(t,y)=\cL_\zs{h_\zs{n-1}}(t,y)
\quad\mbox{for}\quad n\ge 1\,.
\end{equation}
and study the convergence of this sequence in $\cK$. Actually, we will use the existence argument of a fixed point, for a contracting operator in a complete metric space.

\section{Solution of the HJB equation}\label{sec:Solution-HJB}
We give in this section the existence and uniqueness result, of a solution for the HJB equation \eqref{eq:HJB-h}. For this, we show some properties of the Feynman-Kac operator
$\cL_\zs{.}$.

\begin{proposition}\label{Pr.sec:PrL.2} The operator $\cL_\zs{.}$ is "stable" in $\cX$ that is
$$\cL_\zs{f}\in\cX, \quad    \forall \,f \in \cX \,.$$
\end{proposition}

\begin{proof}
Obviously, that for any $f\in\cX$ we have $\cL_\zs{f} \ge 1$. Moreover,
 setting
\begin{equation}\label{sec:PrL.1}
\wt{f}_\zs{s}=f(s,\eta^{t,y}_\zs{s})
\,,
\end{equation}
we represent $\cL_\zs{f}(t,y)$ as
\begin{equation}\label{sec:PrL.2}
\cL_\zs{f}(t,y) =
\E\,\cG(t,T,y)
+\frac{1}{q_\zs{*}}
\, \int_{t}^{T}
\E \left(
\wt{f}_\zs{s}\right) ^{1-q_\zs{*}}\,
\cG(t,s,y) ds\,.
\end{equation}
Therefore, taking into account that $\wt{f}_\zs{s}\ge 1$
and $q_\zs{*}\ge 1$ we get
\begin{equation}\label{sec:PrL.3}
\cL_\zs{f}(t,y)
\,
\le
\,
e^{Q_\zs{*}(T-t)}+\int_{t}^{T}\dfrac{1}{q_\zs{*}}\,
e^{Q_\zs{*}(s-t)}\, \d s
\,\le \,\r^{*}\,\,,
\end{equation}
where the upper bound $\r^{*}\,$ is defined in \eqref{def:norme_in_cX-0}.
Now we have to  show that $\cL_f$ is a uniformly continuous function on $\cK$ for any $f\in\cX$. For any $f \in \cX \bigcap C^{1,1} (\cK)$ we consider equation \eqref{eq:HJB-h}, i.e.
\begin{equation}\label{sec:PrL.4}
\left\{
\begin{array}{rl}
\dfrac{\partial}{\partial t}u(t,y)&+
Q(y) \, u(t,y)
+
\, \alpha \,y \,\dfrac{\partial}{\partial y}u(t,y)\\[5mm]
&
+\dfrac{\beta^{2}}{2}\,\dfrac{\partial^{2}}{\partial y^{2}}u(t,y)
+\dfrac{1}{q_\zs{*}}
\left(\dfrac{1}{f(t,y)}\right)^{q_\zs{*}-1}
=0\,;
\\[5mm]
u(T,y)&=1\,.
\end{array}
\right.
\end{equation}
Setting here $\wt{u}(t,y)=u(T_0+T-t,y)$ we obtain a uniformly parabolic equation
for $\wt{u}$ with initial condition $\wt{u}(T_0,y)=1$. Moreover, we know that $Q$ has bounded
derivative. We deduce that for any $f$ from $\cX \bigcap C^{1,1}(\cK)$, Theorem 5.1 from \cite{LadyzenskajaSolonnikovUralceva1967} (p. 320)
with  $0<l<1$  provides the existence of the unique solution of \eqref{sec:PrL.4}
belonging to $\C^{1,2}(\cK)$. Applying the It\^o formula to the process
$$
\left(u(s,\eta^{t,y}_\zs{s})\,
e^{\int^{s}_\zs{t} Q(\eta^{t,y}_\zs{v})\, \d v}\right)_\zs{t\le s\le T}
$$
 and taking into account equation \eqref{sec:PrL.4} we get
\begin{equation}\label{eq:u=Lf}
u(t,y)=\cL_\zs{f}(t,y)\,.
\end{equation}

\noindent
Therefore, the function
$\cL_\zs{f}(t,y)\in \C^{1,2}(\cK)$, i.e. $\cL_\zs{f}\in\cX$
for any $f\in\cX \bigcap C^{1,1}(\cK)$.

\medskip
\noindent Moreover, for any $f\in \cX$  there exists a sequence $(f_n)_{n\ge 1}$ from $\cX \bigcap C^{1,1}(\cK)$ such that
$$
\sup_{(t,y) \in \cK}|f_n(t,y)-f(t,y)| \to 0 \quad \mbox{as} \quad n \to \infty \,.
$$
This implies
$$
\sup_{(t,y) \in \cK}|\cL_{f_n}(t,y)-\cL_f(t,y)| \to 0 \quad \mbox{as} \quad n \to \infty \,.
$$
So $\cL_f(t,y)$ is uniformly continuous on $\cK$ i.e. $\cL_f \in \cX$. Hence Proposition~\ref{Pr.sec:PrL.2}.
\end{proof}

%
%


\bigskip

\begin{proposition}\label{Pr.sec:Prl.3}
The mapping $\cL$ is a contraction in the metric space
$(\cX , \varrho_\zs{*})$, i.e.
for any $f$, $g$ from $\cX$
\begin{equation}\label{sec:PrL.6}
\varrho_\zs{*}(\cL_\zs{f},\cL_\zs{g})
\le \lambda \varrho_\zs{*}(f,g)\,,
\end{equation}
where 
\begin{equation}\label{def:lambda}
\lambda=\frac{1}{\zeta + 1}, \quad \zeta >0 \,.
\end{equation}\
\end{proposition}
Actually, as shown in Corollary \ref{Corolaire:super-geometric-rate}, an appropriate choice of $\zeta$ gives a super-geometric convergence rate for the sequence $(h_\zs{n})_\zs{n\ge 1}$ defined in \eqref{def:hn-sequence}, to the limit function $h(t,y)$, which is the fixed point of the operator $\cL$.

\bigskip
\begin{proof}
First note that,
for any $f$ and $g$ from $\cX$ and for any $y\in\bbr$
\begin{align*}
|\cL_\zs{f}(t,y)-\cL_\zs{g}(t,y)|
&\le
\frac{1}{q_\zs{*}}\,
 \E\,
\int_{t}^{T}
\cG(t,s,y)
\left\vert
\left(\wt{f}_\zs{s}\right) ^{1-q_\zs{*}}
-
\left(\wt{g}_\zs{s}\right) ^{1-q_\zs{*}}
\right\vert
\, \d s\\[2mm]
&\le
\gamma
\,
\E \,\int^{T}_\zs{t}\,\,\cG(t,s,y)\,\left\vert \wt{f}_\zs{s}-\wt{g}_\zs{s}\right\vert\, \d s  \,.
\end{align*}
We recall that  $\wt{f}_\zs{s}=f(s,\eta^{t,y}_\zs{s})$
and $\wt{g}_\zs{s}=g(s,\eta^{t,y}_\zs{s})$. Taking into account here that
$\cG(t,s,y)\le e^{Q_\zs{*}(s-t)}$ we obtain
$$
|\cL_\zs{f}(t,y)-\cL_\zs{g}(t,y)|
\le
\,
\int^{T}_\zs{t}\,e^{Q_\zs{*}(s-t)}\,\E
|\wt{f}_\zs{s}-\wt{g}_\zs{s}|
\, \d s\,.
$$
Taking into account in the last inequality, that
\begin{equation}\label{sec:PrL.7}
|
\wt{f}_\zs{s}-\wt{g}_\zs{s}
|
\le \,e^{\varkappa (T-s)}\,\varrho_\zs{*}(f,g)
\quad\mbox{a.s.}
\,,
\end{equation}
we get for  all $(t,y)$ in $\cK$
\begin{equation}\label{sec:PrL.8}
\left\vert e^{-\varkappa(T-t)}
\left( \cL_\zs{f}(t,y)-\cL_\zs{g}(t,y)\right) \right\vert
\le
\frac{1}{\varkappa-Q_\zs{*}}\,\varrho_\zs{*}(f,g)
\,.
\end{equation}
Taking into account the definition of $\varkappa$ in
\eqref{def:norme_in_cX},
we obtain inequality \eqref{sec:PrL.6}.
Hence Proposition~\ref{Pr.sec:Prl.3}. \end{proof}

\begin{proposition}\label{Pr.sec:Prl.4}
The fixed point equation $\cL_\zs{h}=h$ has a unique solution in $\cX$.
\end{proposition}
\begin{proof}
Indeed, using the contraction of the operator $\cL$ in
$\cX$ and the definition of the sequence $(h_\zs{n})_\zs{n\ge 1}$
in \eqref{def:hn-sequence} we get, that for any $n\ge 1$
\begin{equation}\label{sec:PrL.10}
\varrho_\zs{*}(h_\zs{n},h_\zs{n-1})\le \lambda^{n-1}\,
\varrho_\zs{*}(h_\zs{1},h_\zs{0})\,,
\end{equation}
i.e. the sequence $(h_\zs{n})_\zs{n\ge 1}$ is fundamental in $(\cX,\varrho_\zs{*})$.
The metric space $(\cX,\varrho_\zs{*})$ is complete since it is included in the Banach space $C^{0,0} (\cK)$, and $\|.\|_\infty$  is equivalent to $\|.\|_*$ defined in \eqref{def:norme_in_cX}.
Therefore, this sequence has a limit in $\cX$, i.e. there exits a function $h$ from $\cX$ for which
$$
\lim_\zs{n\to\infty}\,\varrho_\zs{*}(h,h_\zs{n})=0\,.
$$
Moreover, taking into account that
$h_\zs{n}=\cL_\zs{h_\zs{n-1}}$ we obtain, that
for any $n\ge 1$
$$
\varrho_\zs{*}(h,\cL_\zs{h})
\le
\varrho_\zs{*}(h,h_\zs{n})
+
\varrho_\zs{*}(\cL_\zs{h_\zs{n-1}},\cL_\zs{h})
\le \varrho_\zs{*}(h,h_\zs{n})
+\lambda
\varrho_\zs{*}(h,h_\zs{n-1})\,.
$$
The last expression tends to zero as $n\to\infty$. Therefore
$\varrho_\zs{*}(h,\cL_\zs{h})=0$, i.e. $h=\cL_\zs{h}$ .
Proposition~\ref{Pr.sec:Prl.3} implies immediately that this solution is unique.
\end{proof}

\bigskip
\noindent We are ready to state the result about the solution of the HJB equation.
\begin{proposition}\label{Pr.sec:Prl.5}
The HJB equation \eqref{eq:HJB-h} has a unique solution which is the solution
$h$ of the fixed-point problem $\cL_\zs{h}=h$.
\end{proposition}
\begin{proof}
First, note that in view of Lemma \ref{lemma_4.19*}, the function $\cL_h(t,y)$ is $1/2$-Hl\"{o}lderian with respect to $t$ on $|y|<n$ for any $n \ge 1$. Therefor, choosing in \eqref{sec:PrL.4} the function $f=f_n(t,y)=u(t,\wt y_n)$ (where $\wt y_n$ is the projection of $y$ into $[-n,n]$) we obtain through Theorem 5.1 from \cite{LadyzenskajaSolonnikovUralceva1967} (p. 320) and Lemma \ref{lemma_4.19*}, that the equation \eqref{sec:PrL.4} has a unique solution $u_n(t,y)$. It is clear that the function
$$
u(t,y)=\sum_{n \ge 1} u_n (t,y) \, \textbf{1}_{\{n-1<|y|\le n\}}
$$
is the solution to equation \eqref{sec:PrL.4} for $f=u(t,y)$. Taking into account
the representation \eqref{eq:u=Lf} and the fixed point equation $\cL_\zs{h}=h\,$
we obtain, that the solution of equation \eqref{sec:PrL.4}
$$
u=\cL_\zs{h}=h\,.
$$
Therefore, the function $h$ satisfies equation \eqref{eq:HJB-h}.
Moreover, this solution is unique since $h$ is the unique solution of the fixed point problem.

Choosing in \eqref{sec:PrL.4} the function $f=u$
and taking into account
the representation \eqref{eq:u=Lf} and the fixed point equation $\cL_\zs{h}=h\,$
we obtain, that the solution of equation \eqref{sec:PrL.4}
$$
u=\cL_\zs{h}=h\,.
$$
Therefore, the function $h$ satisfies equation \eqref{eq:HJB-h}.
Moreover, this solution is unique since $h$ is the unique solution
of the fixed point problem.
\end{proof}

\begin{remark}
\begin{enumerate}
  \item We can find in \cite{MaProtterYong1994} an existence and uniquness proof for a more general quasilinear equation but therein, authors did not give a way to calculate this solution, whereas in our case, the solution is the fixed point function for the Feynman-Kac operator.  Moreover our method allows to obtain the super geometric convergence rate for the sequence approximating the solution, which is a very important property in practice. In \cite{Delarue2002} author shows an existence and uniquness result where  the global result is deduced from a local existence and uniqueness theorem.
  \item The application of contraction mapping or fixed-point theorem to solve nonlinear PDE in not new see, e.g. \cite{Freidlin1985} and \cite{Pazy1983} where the term "generalised solution" is used for quasilinear/semilinear PDE, and the fixed point of the Feynamn-Kac representation is discussed.
\end{enumerate}
\end{remark}
\section{Super-geometric convergence rate}\label{sec:geometric_rate}
For the sequence $(h_n)_{n \ge 1}$ defined in \eqref{def:hn-sequence}, and $h$ the fixed point solution for $h=\cL_h$, we study the behavior of the deviation
$$
\Delta_\zs{n}(t,y)=h(t,y)-h_\zs{n}(t,y)\,.
$$

\noindent In the following theorem
we make an appropriate choice of $\zeta$ for the contraction parameter $\lambda$ to get the super-geometric convergence rate for the sequence $(h_\zs{n})_\zs{n\ge 1}$.

\begin{theorem}\label{Th.sec:Mr.1}
The fixed point problem $\cL_h=h$ admits a unique solution $h$ in $\cX$
such that for any $n\ge 1$ and $\zeta>0$
\begin{equation}\label{sec:Mr.10}
\sup_\zs{(t,y)\in \cK}
\,
|\Delta_\zs{n}(t,y)|
\le \B^{*}\,\lambda^{n}\,,
\end{equation}
where
$\B^{*}=e^{\varkappa \wt T}\,(1+\r^{*})/(1-\lambda)$
and $\varkappa$ is given
in \eqref{def:norme_in_cX}.
\end{theorem}

\begin{proof}
Proposition~\ref{Pr.sec:Prl.4} implies
the first part of this theorem. Moreover, from \eqref{sec:PrL.10}
it is easy to see, that
for each $n\ge 1$
$$
\varrho_\zs{*}(h,h_\zs{n})\le \frac{\lambda^{n}}{1-\lambda}\,
\varrho_\zs{*}(h_\zs{1},h_\zs{0})\,.
$$
Thanks to Proposition~\ref{Pr.sec:PrL.2} all the functions
$h_\zs{n}$ belong to $\cX$, i.e.
by the definition of the space $\cX$
$$
\varrho_\zs{*}(h_\zs{1},h_\zs{0})\le
\sup_\zs{(t,y)\in\cK}\,|h_\zs{1}(t,y)-1| \le 1+\r^{*}\,.
$$
Taking into account that
$$
\sup_\zs{(t,y)\in \cK}
\,
|\Delta_\zs{n}(t,y)|
\le e^{\varkappa \wt T}
\varrho_\zs{*}(h,h_\zs{n})\,,
$$
we obtain the inequality \eqref{sec:Mr.10}. Hence
Theorem~\ref{Th.sec:Mr.1}.
\end{proof}

\vspace{3mm}
\noindent Now we can minimize the upper bound \eqref{sec:Mr.10} over $\zeta>0$. Indeed,

$$
\B^{*}\,\lambda^{n}=\C^{*}\,
\exp\{\g_\zs{n}( \zeta)\}\,,
$$
where
$\C^{*}=(1+\r^{*}) \,e^{(Q_\zs{*}+1) \wt T}$
and
$$
\g_\zs{n}(x)=x\,\wt T\,-\ln x-(n-1)\ln(1+x)\,.
$$
Now we minimize this function over $x>0$, i.e.
$$
\min_\zs{x>0} \g_\zs{n}(x)
=
x^{*}_\zs{n}\, \wt T\,-
\ln x^{*}_\zs{n}-(n-1)\ln(1+x^{*}_\zs{n})
\,,
$$
where
$$
x^{*}_\zs{n}=\frac{\sqrt{(\,\wt T\,-n)^2+4\,\wt T\,}+n-\,\wt T\,}{2\,\wt T\,}\,.
$$
Therefore, for
$$
\zeta=\zeta^{*}_\zs{n}=x^{*}_\zs{n}
$$
 we obtain  the optimal upper bound \eqref{sec:Mr.10}.
\begin{corollary}\label{Corolaire:super-geometric-rate}
The fixed point problem has a unique solution $h$ in $\cX$
such that for any $n\ge 1$
\begin{equation}\label{sec:Mr.11}
\sup_\zs{(t,y)\in \cK}
\,
|\Delta_\zs{n}(t,y)|
\,
\le \U^{*}_\zs{n}\,,
\end{equation}
where
$\U^{*}_\zs{n}=\C^{*}\,\exp\{\g^{*}_\zs{n}\}$. Moreover one can check directly that for any $0<\delta<1$
$$
\U^{*}_\zs{n}=\mbox{O}(n^{-\delta n})\quad\mbox{as}\quad n\to\infty\,.
$$
This means that the convergence rate is more rapid than any
geometric one, i.e. it is  super-geometric.
\end{corollary}

\section{Known parameters}\label{sec:Optimal-Strategy}
We consider our optimal consumption and investment problem in the case of markets with known parameters. The next theorem is the analogous of theorem 3.4 in \cite{BerdjanePergamen2012}. The main difference between the two results is that the drift coefficient of the process $Y$ in  \cite{BerdjanePergamen2012} must be bounded and so does not allow the Ornstein-Uhlenbeck process. Moreover the economic factor $Y$ is correlated to the market by the Brownian motion $\U$, which is not the case in the present paper, since we consider the process $\U$ independent of $W$.

\begin{theorem}\label{Theorem_main_1} The optimal value of $J(T_0,x,y,\vartheta)$ for the optimization problem \eqref{eq:Optimisation-Problem}
is given by
$$
J_{T_0}^*=J(T_0,x,y,\vartheta^*)=\sup_\zs{\vartheta\in\cV}\,J(T_0,x,y,\vartheta)=x^\gamma \,h(T_0,y)
$$
where $h(t,y)$ is the unique solution of equation \eqref{eq:HJB-h}.
Moreover, for all $T_0\le t\le T$ an optimal financial strategy
$\vartheta^{*}=(\pi^{*},c^{*})$ is  of the form
\begin{equation}\label{eq:pi*_and_c*-opt-strategy}
\left\{
\begin{array}{rl}
\pi^{*}_\zs{t}=\pi^{*}(Y_\zs{t})
&=\,\dfrac{\theta(Y_\zs{t})}{1-\gamma}\,;
 \\[5mm]
c^{*}_\zs{t}=
c^{*}(t,Y_\zs{t})
&=
\left(h(t,Y_\zs{t})\right)^{-q_\zs{*}}\,.
\end{array}
\right.
\end{equation}
The optimal wealth process $(X^{*}_t)_{T_0\le t\le T}$ satisfies
the following stochastic equation
\begin{equation}\label{eq:EDS_X*}
\d X^{*}_\zs{t}\,=\a^{*}(t,Y_\zs{t})
X^{*}_\zs{t}\, \d t+X^{*}_\zs{t} \b^{*}(Y_\zs{t}) \,\d W_\zs{t}\,, \quad
\  X^{*}_\zs{T_0}\,=\,x\,,
\end{equation}
where
\begin{equation}\label{eq:a*_and_b*}
\left\{
\begin{array}{rl}
\a^{*}(t,y)&=
\dfrac{|\theta(y)|^2}{1-\gamma}+ r\,-\left( h(t,y)\right)^{-q_\zs{*}}\,; \\[4mm]
\b^{*}(y)&=
\dfrac{ \theta(y)}{1-\gamma}\,.
\end{array}
\right.
\end{equation}
The solution $X_t^*$ can be written as
\begin{equation}\label{def:X*_s}
X^{*}_\zs{s}=X^{*}_\zs{t}\,e^{\int^{s}_\zs{t} \a^{*}(v,Y_\zs{v})\, \d v}\,\cE_\zs{t,s}\,,
\end{equation}
where
$
\cE_\zs{t,s}=\exp\left\{\int^{s}_\zs{t} \b^{*}(Y_\zs{v})\, \d W_\zs{v}-
\frac{1}{2}\int^{s}_\zs{t} |\b^{*}(Y_\zs{v})|^{2}\, \d v\right\}\,.
$
\end{theorem}
\vspace{2mm}
\noindent The proof of the theorem follows the same arguments, as Theorem 3.4 in \cite{BerdjanePergamen2012}, so it is omitted.

\section{Unknown parameters}\label{sec:delta-optimal strategy}

In this section we consider the Black-Scholes market with unknown stock price appreciation rate $\mu$ and the unknown drift parameter $\alpha$ of the economic factor $Y$. We observe the process $Y$ in the interval $[0,T_0]$, and use sequential methods to estimate the drift. After that, we will deal with the consumption-investment optimization problem on the finite interval $[T_0,T]$ and look for the behavior of the optimal value function $J^*(T_0,x,y)$ under the estimated parameters. We define the value function $\wh J^*_{T_0}$ the estimate of $J^*_{T_0}$
\begin{eqnarray}\label{def:wh J*}
\wh J^*_{T_0}&:=&\E_{T_0} \,\left(\int^T_\zs{T_0}
(\wh c^*_\zs{t})^{\gamma}\,(\wh X^{*}_\zs{t})^{\gamma}\, \d t\,+\,(\wh X^{*}_\zs{T})^{\gamma} \, \right).
\end{eqnarray}
$\E_{T_0}$ is the conditional expectation $ \E(\,.\,\vert \cF_{T_0})$. $\wh X^{*}_\zs{t}$ is a simplified notation for $ X^{\wh \vartheta_*}_\zs{t}$ and from \ref{def:X*_s} we write
\begin{equation}\label{def:wh-X*_s}
\wh X^{*}_\zs{s}=\wh X^{*}_\zs{t}\,e^{\int^{s}_\zs{t} \wh \a^{*}(v,Y_\zs{v})\, \d v}\,\wh \cE_\zs{t,s}\,,
\end{equation}
where
$
\wh \cE_\zs{t,s}=\exp\left\{\int^{s}_\zs{t} \wh \b^{*}(Y_\zs{v})\, \d W_\zs{v}-
\frac{1}{2}\int^{s}_\zs{t} |\wh \b^{*}(Y_\zs{v})|^{2}\, \d v\right\}\,.
$
The functions $\a^{*}(t,y)$ and $\b^{*}(t,y)$ are defined as
\begin{equation}\label{eq:wh-a*_and_wh-b*}
\left\{
\begin{array}{rl}
\wh \a^{*}(t,y)&=
\dfrac{|\wh \theta(y)|^2}{1-\gamma}+ r\,-\left( \wh h(t,y)\right)^{-q_\zs{*}}\,; \\[4mm]
\wh \b^{*}(y)&=
\dfrac{ \wh \theta(y)}{1-\gamma}\,, \quad
\wh \theta(y)=\dfrac{\wh \mu-r}{\sigma(y)}\,.
\end{array}
\right.
\end{equation}
The estimated consumption process is $\wh c^*_\zs{t}=\wh c^{*}(t,Y_\zs{t})=\left(\wh h(t,Y_\zs{t})\right)^{-q_\zs{*}}\,$ and $\wh h(t,y)$ is the unique solution for $h=\wh \cL_{h}$. The operator $\wh \cL$ is defined by:
\begin{equation}\label{def:wh-L-theOperator}
\wh \cL_\zs{f}(t,y)= \E\,\wh \cG(t,T,y)
+\frac{1}{q_\zs{*}}
\int_{t}^{T}
\E \left( (f(s,\wh \eta_s^{t,y}))^{1-q_*}\,\wh \cG(t,s,y)\right)
\,\d s\,,
\end{equation}
where $\wh \cG(t,s,y)= \exp \left( \int_{t}^{s} \wh Q(\wh \eta_{u}^{t,y})\, \d u\right)$. The process $(\wh \eta^{t,y}_s)_\zs{ t\le  s\le T}$ has the following dynamics:
\begin{equation}\label{def:EDS-wh-eta}
\d \wh \eta^{t,y}_\zs{s}=
\wh \alpha \wh \eta^{t,y}_\zs{s} \, \d s +\beta \, \d \wt \U_\zs{s}, \quad \quad
\wh \eta^{t,y}_\zs{t}=y\,.
\end{equation}
Here $\wh \alpha$ and $\wh \mu$ are some estimates for the parameters $\alpha$ and $\mu$ which will be specified
later.

\subsection{Sequential procedure}
We assume the unknown parameter $\alpha$ taking values in some bounded interval $[\alpha_2,\alpha_1]$, with $\alpha_2<\alpha<\alpha_1 < 0$. 
We define
$\wh \alpha$ as the projection onto the interval $[\alpha_2,\alpha_1]$ of the sequential estimate $\alpha^*$.
\begin{equation}\label{def:alpha_trunc_seq_est}
\wh \alpha=Proj_{[\alpha_2,\alpha_1]} \alpha^*, \quad\quad \alpha^*= \left( \frac{\int_0^{\tau_H} Y_t \, \d Y_t}{H} \right) 1_{\{ \tau_H \le T_0 \}}
\end{equation}
where $\tau_H=\inf \left\{ t \ge 0, \int_0^{t } Y_s^2 \, \d s \ge H \right\}$.
Furthermore, 
we introduce  the function $\e \, (\,.\,)$, which will serve later for the $\delta$-optimality:

\begin{equation}\label{def:e(T0)}
\e \,(T_0)=\sqrt{\dfrac{\beta^2}{H}+\dfrac{\alpha_2^2 }{\beta^{12}}\left( \dfrac{\k(3) }{T_0^{2}} \right)}.
\end{equation}
Here $H=\beta_2\,(T_0-T^{\ve }_0)$, $\beta_2=\beta^2/2 |\alpha_2|$, $\ve  =5/6$ and
\begin{equation*}
\k(m)=3^{2 m-1}\,\left( Y_0^{2 m}+(1+(m (2 m-1))^m\,( 2 \,\beta)^{2m})\,\k_1 (m) \right),
\end{equation*}
with $\k_1 (m)=2^{2 m-1} \left( Y_0^{2 m}+(2m-1)!! \,\beta_1^m \right) $ and  $\beta_1=\beta^2/2 |\alpha_1|$.
The proposition bellow gives $ \wh \alpha$ the truncated sequential estimate of $\alpha$ and gives a bound for the expected deviation
$\E |\ov \alpha$, where $\ov \alpha =\wh \alpha- \alpha$.

\begin{proposition}\label{prop:sequential_alpha}
For any $0<T_\zs{0}<T$
$$
\E |\ov \alpha| \le \e \,(T_0)\,.
$$

\end{proposition}

\begin{proof}
Note first that $\E |\ov \alpha| \le \E |\alpha^*-\alpha|$, so it is enough to show that $\E |\alpha^* -\alpha| \le \e \,(T_0)$. Moreover, we know from \cite{LiptserShiryaev2000-II} chapter 17, that the maximum likelihood estimate of $\alpha$ is given by
$$
\frac{\int_0^{T_0} Y_t \, \d Y_t}{\int_0^{T_0 } Y_t^2 \, \d t }\,.
$$
To estimate $\alpha$ we use the sequential maximum likelihood estimate proposed in
\cite{LiptserShiryaev2000-II} and \cite{Novikov1971}
$$
\wt \alpha=\frac{\int_0^{\tau_H} Y_t \, \d Y_t}{\int_0^{\tau_H } Y_t^2 \, \d t }=\alpha+\beta \,\frac{\int_0^{\tau_H} Y_t \, \d \U_t}{H}\,.
$$
Taking into account that  $\int_0^{\infty } Y_t^2 \, \d t =+\infty$ a.s.,
we obtain that  $\wt \alpha  \rightsquigarrow \cN (\alpha,  \beta^2/H )$ and hence
$
\E \,|\wt \alpha-\alpha|^2 = \beta^2/H.
$
 The problem with the previous estimate is that $\tau_H$ may be greater than $T_0$. To overcome this difficulty
 we use the truncated modification of the sequential estimate $\wh \alpha$ from \cite{PergamenKonev92}, i.e. $\alpha^* =\wt \alpha \, 1_{\{ \tau_H \le T_0 \}}$. We observe that
\begin{eqnarray*}
 \alpha^* -\alpha&=&(\alpha^*-\alpha) 1_{\{ \tau_H \le T_0 \}}+(\alpha^*-\alpha) 1_{\{ \tau_H > T_0 \}}\\[2mm]
&=& \beta \,\frac{\int_0^{\tau_H} Y_t \, \d \U_t}{H} 1_{\{ \tau_H \le T_0 \}}- \alpha 1_{\{ \tau_H > T_0 \}}\,.
\end{eqnarray*}
So
\begin{eqnarray}\label{equ:inequalite-E(wh alpha-alpha)2}
\E( \alpha^* -\alpha)^2 &=\dfrac{\beta^2}{H^2}&\E\left(  \int_0^{\tau_H} Y_t \, \d \U_t  \,  1_{( \tau_H \le T_0 )}\right)^2+ \alpha^2 \,\P( \tau_H > T_0 )\nonumber\\
&\le& \frac{\beta^2}{H^2}\,\E \left(  \int_0^{\tau_H} Y_t \, \d \U_t \right)^2+\alpha^2 \P( \tau_H > T_0 )\nonumber\\[2mm]
&\le& \frac{\beta^2}{H}+\alpha^2 \,\P( \int_0^{T_0} Y_t^2 dt < H )\,.
\end{eqnarray}
Moreover, by the It\^o formula
$$
\d Y_t^2=2 Y_t \, \d Y_t+\beta^2 \, \d t=(2\alpha Y_t^2+\beta^2) \,\d t+2\beta Y_t \, \d \U_t\,.
$$
From there we deduce that
$$
\int_0^{T_0}(2\alpha Y_t^2+\beta^2)\, \d t= Y_{T_0}^2-Y_0^2-2 \beta \int_0^{T_0} Y_t \, \d \U_t \,.
$$
Taking into account that $\alpha_2 \le \alpha \le \alpha_1 <0$ and using the Markov's inequality, we get for any integer $m > 0$
\begin{eqnarray*}
 \P \left( \int_0^{T_0} Y_t^2 dt < H \right)&=&\P \left( \int_0^{T_0} (2\alpha Y_t^2+\beta^2)dt > 2 \alpha H+\beta^2 T_0 \right)\\
&=&\P \left( Y_{T_0}^2-Y_0^2- 2\beta \int_0^{T_0} Y_t \, \d \U_t >2 \alpha H+\beta^2\, T_0 )\right)\\
&\le&\frac{\E \left( Y_{T_0}^2-Y_0^2-2 \beta \int_0^{T_0} Y_t \, \d \U_t \right)^{2 m } }{(2 \alpha_2 H+\beta^2 T_0)^{2 m }}\,.
\end{eqnarray*}
Here $2 \alpha H+\beta^2 \, T_0 >0$, ie: $0< H < \beta_2\, T_0$. For the stochastic integral $\xi_t=\int_{0}^t \beta e^{\alpha (t-v)} \, \d \U_v$ we get for any  $m \in \bbn_*$
$$
\E(\xi_t^{2 m})=(2m-1)!! \,[\E(\xi^2_t)]^m \le (2m-1)!! \, \beta^m_1.
$$
Furthermore, in view of $Y_{T_0}=Y_0\, e^{\alpha\,T_0}+\xi_\zs{T_0}$ we obtain
$$
\E Y_{T_0}^{2 m} \le 2^{2 m-1} \left( \E(Y_0 e^{\alpha T_0})^{2 m}+ \E (\xi_\zs{T_0}^{2 m}) \right)
\le \k_1 (m)\,.
$$
Moreover, we have (see e.g. \cite{LiptserShiryaev2000-I} Lemma 4.12):
$$
 \E \left(\int_0^{T_0} Y_t \, \d \U_t \right)^{2 m} \le (m (2 m-1))^m \, T_0^{m-1} \int_0^{T_0} \E Y_s^{2 m} \, \d s
\le  \k_2 (m) \, T^m_0\,.
$$
where $\k_2 (m)=(m (2 m-1))^m \, \k_1(m)$ . We conclude that
\begin{eqnarray*}
 \P \left( \int_0^{T_0} Y_t^2 dt < H \right)&\le&\frac{3^{2 m-1}\,\left( Y_0^{2 m}+\k_1 (m)+ ( 2 \,\beta)^{2m}\, \k_2(m)\, T_0^m \right)}{(2 \alpha_2 H+\beta^2 T_0)^{2 \, m }} \,.
\end{eqnarray*}
We set $H=\beta_2 \, (T_0-T_0^\ve )$ for some $\ve $, we obtain
\begin{eqnarray*}
 \P \left( \int_0^{T_0} Y_t^2 dt < H \right)&\le& \dfrac{1 }{(\beta^2)^{2 m}}\left( \dfrac{\k (m) }{T_0^{ m \,(2\,\ve -1)}} \right)\,.
\end{eqnarray*}
Replacement in \eqref{equ:inequalite-E(wh alpha-alpha)2} gives
$$
\E \, (\alpha^*-\alpha)^2 \le \dfrac{\beta^2}{\beta_2\,(T_0-T^{\ve }_0)}+\dfrac{\alpha^2 }{\beta^{4 \,m}}\left( \dfrac{\k (m) }{T_0^{ m \,(2\,\ve -1)}} \right)\,.
$$
We fixe $\ve =5/6$ and $m=3$ so that $ m \,(2 \, \ve -1) = 2$,  which gives $\e^2 \,(T_0)$ and then the desired result.
\end{proof}

\subsection{Known stock price appreciation rate $\mu$}

We consider in this section the consumption-investment problem for markets with known $\mu$ and unknown $\alpha$, i.e. in this case
$\wh \mu=\mu$ and, therefore,
$\wh \theta(y)=\theta(y)$
in \eqref{eq:wh-a*_and_wh-b*}. To state the approximation result we set
\begin{equation}\label{def:h_T-Gamma_T}
\left\{
\begin{array}{rl}
\h_{1}&=\dfrac{1+2\,\gamma+\zeta_0}{1+\zeta_0} \,\dfrac{\wt T}{|\alpha_1|}\,\left(2\, Q_1^*\wt T  + \gamma\, h_1^* \, \right)  , \\[4mm]
\Gamma \,&=\,\left(q_*\,\wt T(\wt d )^\gamma+ \, (\wt T+1)\, \left(\sqrt{\wt c q_* }\right)^\gamma \right)
\dfrac{1}{\varkappa^\gamma} \,e^{\gamma \,\varkappa \,\wt T}\,.
\end{array}
\right.
\end{equation}
Here $\zeta_0 >0$,  $\wt c=4\, \wt T\, e^{c_0 \wt T} \wt d^2$,
$
c_0 =2 \, \sup_{(s,y)\in \cK} (|\a^{*}(s,y)|^2+|\b^{*}(s)|^2).
$
Moreover, $\wt d\,$ is the upper bound \eqref{def:wt d} and
\begin{equation}\label{def:h*-1}
h_1^*\,=\left( \wt T Q_\zs{1}^*+\frac{Q_1^* \wt T^2}{q_\zs{*}}\right) \,e^{Q_\zs{*}\,\wt  T}\\
+\frac{3}{q_\zs{*}}\,\sqrt{\dfrac{2 |\alpha_2|}{\beta^2(1-e^{2 \alpha_2})}} e^{Q_*\,\wt  T}\, \wt T\,.
\end{equation}

\vspace{3mm}
\noindent We notice that int the estimation interval $[0,T_0]$, we don't invest in the risky stock. We chose the strategy $(c_t,\pi_t)=(r,0)$ for $0\le t \le T_0$, so that $\forall \,0 \le t \le T_0,\, \, X_{t}=X_0=x$ a.s.

\vspace{2mm}
\begin{theorem} \label{theorem:wh J*-J*} For any  $0< T_0 < T$ and any $m \ge 1$
\begin{equation}\label{theorem:wh J*-J*-eq-1}
\E\,|\wh J^*_{T_0} -J^*(T_0,x,Y_{T_0})| \le \,  \delta_\zs{m} \,,
\end{equation}
where
$$
 \delta_\zs{m}=\delta_\zs{m}(x,T_0)=\Gamma \,\h^\gamma_{1}\,x^\gamma\,\left(  \big(2\,\iota_0\big)^\gamma\,+
 \left(\,(2m-1)!! \, \beta^{2\,m}/(2 |\alpha_1|)^{m}\right)^{\gamma/2\,m}\, \right) \,\e\,(T_0)^\gamma\,,
$$
 $\iota_0=\beta/\sqrt{2\,|\alpha_1|}$ and $\e \,(T_0)$ is defined in \eqref{def:e(T0)}.
\end{theorem}

\begin{proof}
Note that for any $T_0 <T$
\begin{eqnarray}
|\wh J^*_{T_0}-J^*_{T_0}|&\le& \E_\zs{T_0}\,\left(\int^T_\zs{T_0}
|(\wh c^*_\zs{t})^{\gamma}\,(\wh X^{*}_\zs{t})^{\gamma}-( c^*_\zs{t})^{\gamma}\,( X^{*}_\zs{t})^{\gamma} |\, \d t \right)\nonumber
+
\E_\zs{T_0} \, |(\wh X^{*}_\zs{T})^{\gamma} -( X^{*}_\zs{T})^{\gamma}|\nonumber\\
&\le&
\E_\zs{T_0} \, \left(\int^T_\zs{T_0}
| \wh c^*_\zs{t}\,\wh X^{*}_\zs{t}- c^*_\zs{t}\, X^{*}_\zs{t} |^{\gamma}\, \d t \right)\nonumber
+
\E_\zs{T_0} \,  | \wh X^{*}_\zs{T} -X^{*}_\zs{T}|^{\gamma}\,.
\end{eqnarray}
Using  Lemma \ref{lemma:XT_bound}  we get
\begin{equation*}
|\wh J^*_{T_0} -J^*(T_0,x,Y_{T_0})| \le \,  \Gamma \,\h_{1}^\gamma\,x^\gamma \,  \big(2\,\iota_0+|Y_{T_0}| \big)^\gamma\,|\ov \alpha|^\gamma
\end{equation*}
\vspace{2mm}
and, therefore,
$$
\E \,|\wh J^*_{T_0} -J^*(T_0,x,Y_{T_0})| \le \Gamma \,x^\gamma \,\h^\gamma_{1}\, \big(2\,\iota_0\big)^\gamma\,\E \,|\ov \alpha|^\gamma
\Gamma \,x^\gamma \,\h^\gamma_{1}\, \E \,\left( \big|Y_{T_0}\big|^\gamma\,|\ov \alpha|^\gamma \right)\,.
$$
By Holder's and Jensen's inequalities for $m' =m\, (2-\gamma)/\gamma>1$ with $m \ge 1$
\begin{eqnarray*}
 \E \,\left( \big|Y_{T_0}\big|^\gamma\,|\ov \alpha|^\gamma \right) &\le&  \left(\E \, \big|Y_{T_0}\big|^{\frac{2 \gamma}{2-\gamma}} \right)^{(2-\gamma)/2} \, \left(\E \,  |\ov \alpha|^{2}\right)^{\gamma/2}\\
&\le&
\left(\E \, \big|Y_{T_0}\big|^{\frac{2\, \gamma\, m'}{2-\gamma}} \right)^{(2-\gamma)/2\,m'} \, \e \,(T_0)^\gamma
\le
\left(\E \, \big|Y_{T_0}\big|^{2\,m} \right)^{\gamma/2\,m} \, \e \,(T_0)^\gamma\,.
\end{eqnarray*}
From \cite{KabanovPergamenshchikov2003}, Lemma 1.1.1 it follows that $\E \, \big|Y_{T_0}\big|^{2m} \le c_m(T_0) \le c_m(0)$  where
$$
c_m(T_0)=(2m-1)!! \, \beta^{2\,m}\, \left(\dfrac{1-e^{2\alpha T_0}}{2 |\alpha|}\right)^{m}\,.
$$
We conclude that for any $m \ge 1$
\begin{eqnarray}\label{bound:E-Ygamma-alphagamma}
 \E \,\left( \big|Y_{T_0}\big|^\gamma\,|\ov \alpha|^\gamma \right)
&\le&
 \, \left(\,(2m-1)!! \, \beta^{2\,m}/(2 |\alpha_1|)^{m}\right)^{\gamma/2\,m}\ \, \e^\gamma \,(T_0)\,,
\end{eqnarray}
which gives the desired result.
 \end{proof}

\begin{remark}
We observe in Theorem \ref{theorem:wh J*-J*}, that the expected deviation $\E\,|\wh J^*_{T_0} -J^*(T_0,x,y)|$ can be arbitrary small, if either we observe the process $Y$ in a wide interval $[0, T_0]$  so that $\E \,|\ov \alpha|$ be small enough, or we invest a small capital $x$ at the initial time. That means, when the estimation interval is not wide enough, which is the case in practice, we can always find a consumption-investment strategy that belongs closer to the optimal one. For this aim, we need to be cautious in choosing the initial endowment and we need to take into account 
the upper bound
\eqref{theorem:wh J*-J*-eq-1}.
\end{remark}

\begin{lemma}\label{lemma:XT_bound}
For any  $0<T_0\le T$
\begin{equation}\label{def:wt d}
\E_\zs{T_0} \left(\sup_\zs{T_0\le s\le T}\,(\wh X^{*}_\zs{s})^2 \right)< x^2 \, \wt d^2
\quad\mbox{and}\quad \wt d^2=4 e^{2\wt T (A^*+(B^*)^2)}\,,
 \end{equation}
where $A^*=\sup_{(s,y)\in \cK} \wh \a^*(s,y)$ and  $B^*=\sup_{(s,y)\in \cK}  \wh \b^*(s,y)$.  Moreover,
\begin{equation}\label{eq:lemma:XT_bound-2}
 \sup_{T_0 \le t \le T}\E_\zs{T_0} |\wh X^*_t- X^*_t|^\gamma \le k_1\, x^\gamma \, \left(\h_{1}\,(2\,\iota_0+|Y_{T_0}| )\right)^\gamma\,|\ov \alpha|^\gamma
 \end{equation}
 and
 \begin{equation}\label{eq:lemma:XT_bound-3}
  \E_\zs{T_0} \,\left(\int^T_\zs{T_0}| \wh c^*_\zs{t}\,\wh X^{*}_\zs{t}- c^*_\zs{t}\,
X^{*}_\zs{t} |^{\gamma}\, \d t \right)
  \le
  k_2\,x^\gamma\, \left(\h_{1}\,(2\,\iota_0+|Y_{T_0}| )\right)^\gamma\,|\ov \alpha|^\gamma\,,
 \end{equation}
where
$k_1=  (\sqrt{\wt c q_*})^{\gamma}\,e^{\gamma \,\varkappa \,\wt T}/\varkappa^\gamma$ and
$k_2=\left(\wt T(\sqrt{\wt c q_* })^{\gamma}+\wt d^\gamma q_* \, \wt T\right) \,e^{\gamma \,\varkappa \,\wt T}/\varkappa^\gamma$.
\end{lemma}

\begin{proof}
It is clear from \eqref{def:wh-X*_s}, that for the bounded function $\wh \b^{*}(y)$ the process
$(\wh \cE_\zs{t,s})_\zs{t\le s\le T}$ is a quadratic integrable martingale and
by the Doob inequality
\begin{eqnarray*}
\E_\zs{T_0} \left(\sup_\zs{T_0\le s\le T}\,(\wh X^{*}_\zs{s})^2 \right) &\le&
x^2 \,e^{2 \wt T A^*} \E\,\sup_\zs{t\le s\le T}\,\wh \cE^{2}_\zs{t,s}\,\le \,x^2\,4\,e^{2 \wt T A^*}\,\E\,\,\wh \cE^{2}_\zs{t,T}\\[2mm]
&\le&  4\, x^2\,e^{2 \wt T A^*}\,e^{ \wt T (B^*)^2}\,.
\end{eqnarray*}
This gives \eqref{def:wt d}. We set $\Delta_t=\wh X^{*}_\zs{t}- X^{*}_\zs{t}$, $A_s=\a^{*}(s,Y_\zs{s})$, $B_s=\b^{*}(Y_\zs{s})$,
$\wh A_s=\wh \a^{*}(s,Y_\zs{s})$ and $\wh B_s=\wh \b^{*}(Y_\zs{s})$.  The functions $\wh \a^{*}(s,y)$ and $\wh \b^{*}(y)$
are defined in \eqref{eq:wh-a*_and_wh-b*}. It is clear that if $\mu$ is known, (i.e. $\wh \mu=\mu$) the function $\wh B_s=B_s$. But we keep this
function to use this proof in the case when the paramster $\mu$ is unknown.
 Moreover we define $\varphi_1(s)=\wh A_s \wh X_s^* - A_s X_s^* $ and $\varphi_2(s)=\wh B_s \wh X_s^* - B_s X_s^* $. So, from \eqref{eq:EDS_X*} we get
\begin{eqnarray*}
\Delta_t^2 &=& \left( \int_{T_0}^t \varphi_1(s) \, \d s+\int_{T_0}^t \varphi_2(s) \, \d W_\zs{s}\,\right)^2
\le 2 (t-T_0)\int_{T_0}^t \varphi_1^2(s) \, \d s+2 \left(\int_{T_0}^t \varphi_2(s) \, \d W_\zs{s}\,\right)^2\,.
\end{eqnarray*}
We observe that
$$
 \varphi^2_1(s) \le \left( |\wh A_s - A_s |\,|\wh X_s^*|+|A_s| |\Delta_s|\right)^2
\le 2 |\wh A_s - A_s |^2\,|\wh X_s^*|^2+ 2|A_s|^2 |\Delta_s|^2 \,.
$$
Furthermore, since $\wh B_s = B_s$ we obtain
$$
 \varphi^2_2(s)\le \left( |\wh B_s - B_s |\,|\wh X_s^*|+|B_s| |\Delta_s|\right)^2\le  |B_s|^2 |\Delta_s|^2 \,.
$$
Setting now $g(t)=\E_\zs{T_0}(\Delta_t^2)$ we obtain
$$
 g(t) \le c_0 \int_{T_0}^t g(s) \, \d s + \psi(t)
 \quad\mbox{and}\quad
\psi(t)= 4\,\wt T\, \int_{T_0}^t \E_\zs{T_0} |\wh A_s - A_s |^2\,|\wh X_s^*|^2 \, \d s\,.
$$
The Gronwall-Bellman inequality yields
\begin{eqnarray*}
 g(t) &\le& \psi(t) e^{c_0 t}
\le x^2\,4 \, \wt T \,e^{c_0 T} \int_{T_0}^t \E_\zs{T_0} \left(|\wh A_s - A_s |^2\,|\wh X_s^*|^2\right) \, \d s\\
&\le& \wt c \, x^2\, \int_{T_0}^t \E_\zs{T_0} |\wh A_s - A_s |^2\,\d s
\le \wt c \, x^2\,\int_{T_0}^t \E_\zs{T_0} |\wh h(s,Y_s)^{-q_*}-h(s,Y_s)^{-q_*} |^2\, \d s\\
&\le& \wt c \, x^2\, q_* \int_{T_0}^t \E_\zs{T_0} |\wh h(s,Y_s)-h(s,Y_s)|^2 \, \d s\,,
\quad
\wt c=4\, \wt T\, e^{c_0 \wt T} \wt d^2\,.
\end{eqnarray*}
Using \eqref{eq:varrho_h-h} and Lemma \ref{lemme:bound-EY-and-EG(t,y)}
we obtain, that for any $T_0 \le s\le T$
\begin{align}\nonumber
\E_\zs{T_0}|\wh h(s,Y_s)&-h(s,Y_s)|
\le \h_{1}\,\E_\zs{T_0} \left( e^{\varkappa (T-s)}(\iota_0+|Y_s|)\right) |\ov \alpha|
\\[2mm] \nonumber
&\le \h_{1}\,(\iota_0+\E_\zs{T_0}|Y_s| ) \, e^{\varkappa (T-s)}\, |\ov \alpha|
\\[2mm]\label{eq:bond_wh_h-h}
&\le \h_{1}\,(2\,\iota_0+|Y_{T_0}| ) \, e^{\varkappa (T-s)}\, |\ov \alpha|\,.
\end{align}
Therefore,
\begin{equation*}
 g(t) \le x^2\,\wt c\, q_* \, \left(\h_{1}\,(2\,\iota_0+|Y_{T_0}| )\right)^2\,\frac{ e^{2\,\varkappa\, \wt T}}{\varkappa^2}\, |\ov \alpha|^2\,.
\end{equation*}
Hence, \eqref{eq:lemma:XT_bound-2} holds. We show now inequality \eqref{eq:lemma:XT_bound-3}.  Note that
in view of
\eqref{eq:pi*_and_c*-opt-strategy} the optimal consumption $0\le c^*_\zs{t}\le 1$. Thus,
\begin{eqnarray*}
\E_{T_0}\,\left(\int^T_\zs{T_0}
| \wh c^*_\zs{t}\,\wh X^{*}_\zs{t}- c^*_\zs{t}\, X^{*}_\zs{t} |^{\gamma}\, \d t \right)
&\le&
\E_{T_0}\,\left(\int^T_\zs{T_0}
|\wh c^*_\zs{t}\,- c^*_\zs{t} |^{\gamma} |\wh X^{*}_\zs{t}|^\gamma\, \d t \right)
+
\,\int^T_\zs{T_0}
\E_{T_0}|\wh X^{*}_\zs{t}- X^{*}_\zs{t} |^{\gamma}\, \d t \\
&\le&
x^\gamma \,\wt d^\gamma \, \,
\E_{T_0}\,\left(\int^T_\zs{T_0}
|\wh c^*_\zs{t}\,- c^*_\zs{t} |^{\gamma} \, \d t \right)
+
\wt T\,\sup_{T_0 \le t \le T}\E_{T_0}\,|\wh X^{*}_\zs{t}- X^{*}_\zs{t} |^{\gamma}\,.
\end{eqnarray*}
Using now the upper bound
 \eqref{eq:bond_wh_h-h} and taking into account that   $\inf_\zs{(t,y) \in \cK}\,h(t,y) \ge 1$,
we obtain
\begin{align*}
 \E_{T_0}\,\left(\int^T_\zs{T_0}|\wh c^*_\zs{t}\,- c^*_\zs{t} |^{\gamma} \, \d t \right)
&\le
q_*\,\,\int^T_\zs{T_0}\E_{T_0}\, |\wh h(s,Y_s)-h(s,Y_s)|^\gamma|\, \d t\\[2mm]
&\le
q_* \, \wt T \,
 \left(\h_{1}\,(2\,\iota_0+|Y_{T_0}| ) \right)^\gamma \,\frac{ e^{\gamma\,\varkappa\, \wt T}}{\varkappa^\gamma}\,\E_{T_0}\,|\ov \alpha|^\gamma\,.
\end{align*}
Therefore
$$
\E_{T_0}\,\left(\int^T_\zs{T_0}
| \wh c^*_\zs{t}\,\wh X^{*}_\zs{t}- c^*_\zs{t}\, X^{*}_\zs{t} |^{\gamma}\, \d t \right)
\le
 k_2\,x^\gamma \, \left(\h_{1}\,(2\,\iota_0+|Y_{T_0}| )\right)^\gamma\, \E_{T_0} \,|\ov \alpha|^\gamma\,.
$$
This implies
 \eqref{eq:lemma:XT_bound-3} and then Lemma \ref{lemma:XT_bound}.
\end{proof}

\subsection{Unknown stock price appreciation rate $\mu$}\label{sec:mu-unknown}
In practice, it is not realistic to consider known the stock price appreciation rate $\mu $. In this section, in addition to the unknown drift parameter $\alpha$ of the economic factor process, we consider an unknown stock price appreciation rate  
$\mu$.
 We recall that the dynamics of the risky stock is given in \eqref{eq:BS-model}. Let $\wh \mu$ its estimate defined by
\begin{equation}\label{def:mu-estimate}
\wh \mu=\dfrac{Z_{T_0}}{T_0} \quad \mbox{ with}\quad Z_t=\int_0^t\frac{1}{S_t} \, \d S_t\,.
\end{equation}

\begin{lemma} \label{lemme:bond-estimated-mu}
For any $0<T_\zs{0}<T$
\begin{equation}\label{def:e(T0)-mu}
\E \vert \wh \mu-\mu \vert \le \e_{1}(T_\zs{0})
\quad\mbox{and}\quad
\E \vert \wh \mu-\mu \vert^{2} \le \e^{2}_{1}(T_\zs{0})
\,,
 \end{equation}
where $\e_{1}(T_0)=\sigma^{*}/\sqrt{T_\zs{0}}$ and $\sigma^*=\sup_{y \in \bbr} \sigma(y)$.
\end{lemma}
\begin{proof}
From the definition of the process $Z$ we get
$$
\wh \mu-\mu=\dfrac{1}{{T_0}} \int_0^{T_0} \sigma(Y_t) \, \d W_t\,.
$$
This implies directly the bounds \eqref{def:e(T0)-mu}.
\end{proof}

\vspace{3mm}
\noindent Let the optimal value functions
$J^*(T_0,x,y)$ and $\wh J^*_{T_0}$ its estimate given in \eqref{def:wh J*}, and let define the constants
$$
k'_1= 2\,\sqrt{\wt c \wt T}\, \left( \dfrac{2 \mu_2+ r+\sigma_1+1}{\sigma_1^2 (1-\gamma)}\right)
\quad \mbox{and}\quad
k'_2=  \frac{e^{\,\varkappa \wt T}}{\varkappa}\,.
$$
Moreover, we define
$$
\Gamma_1=k_3+k_5 \quad \mbox{and} \quad \Gamma_2=k_4+k_6\,,
$$
where
$
k_3=(k'_1)^\gamma+ \left(\sqrt{2\,\wt c\, q_*} \, k'_2\, \h_2 \right)^\gamma,
\quad
k_4= \left(\sqrt{2\,\wt c\, q_*} \, k'_2\, \h_1 \right)^\gamma
$
,
$$
k_5=\wt T\,(k'_1)^\gamma+k_7\, (k'_2 \, \h_2)^\gamma,
\,
k_6=k_7 \,\left(\,k'_2 \, \h_1\right)^\gamma,
\,
k_7=\left(\sqrt{2\,\wt c\, q_*}+q_* \wt d^\gamma \right)\,.
$$
recall that $\wt c=4 e^{c_0 t} \wt d^2$ and $\wt d$ is given in \eqref{def:wt d}. The constants $\h_{1}$ is given in \eqref{def:h_T-Gamma_T} and
\begin{equation}\label{def:h2}
\h_2 \, =\frac{\gamma\, (\mu_2+\,r\,)}{(1-\gamma)\, \sigma^2_1}\frac{2\,\wt T^2}{\iota_0} \,.
\end{equation}
We are dealing with the following result
\begin{theorem} \label{theorem:wh J*-J*-2}
The  estimate of optimal cost function $\wh J^*_{T_0}$ satisfies the following inequalities
\begin{equation}\label{theorem:wh J*-J*-2-eq1}
|\wh J^*_{T_0} -J^*(T_0,x,Y_{T_0})| \le x^\gamma\, \Gamma_1\, (2 \, \iota_0+ |Y_{T_0}|)^\gamma \,  \varpi^\gamma
+x^\gamma \,\Gamma_2 \,(2 \, \iota_0+ |Y_{T_0}|)^\gamma \,|\ov \alpha|^\gamma \,,
\end{equation}
where $\varpi=|\wh \mu-\mu |+|\wh \mu-\mu |^{2}$.
Moreover, for any $m \ge 1$
\begin{equation}\label{theorem:wh J*-J*-2-eq2}
\E\,|\wh J^*_{T_0} -J^*(T_0,x,Y_{T_0})| \le \check{\delta}_\zs{m} \,,
\end{equation}
with $\check{\delta}_\zs{m}=\check{\delta}_\zs{m}(x,T_0)=x^\gamma\,\left(\Gamma_1(\, 3 \iota_0^\gamma+|Y_0|^\gamma) \, \e_{2}(T_0)^\gamma+ \,\check{\Gamma}_\zs{m} \,\e(T_0)^\gamma\right)$ and
$$
\check{\Gamma}_\zs{m}=\Gamma_2\,
\left(
(2 \iota_0)^\gamma+
\left(\,(2m-1)!! \, \beta^{2\,m}/(2 |\alpha_1|)^{m}\right)^{\gamma/2\,m}
\right)\,.
$$
Here  $\iota_0=\beta/\sqrt{2\,|\alpha_1|}$, $\e_{2} \,(T_0)=\e_{1} \,(T_0)+\e^{2}_{1} \,(T_0)$, $\e_{1} \,(T_0)$ is given in \eqref{def:e(T0)-mu} and  $\e \,(T_0)$ is defined in \eqref{def:e(T0)}.
\end{theorem}
\begin{proof}
We follow the same arguments as in the proof of Theorem \ref{theorem:wh J*-J*}, and use Lemma \ref{lemma:XT_bound-2} bellow to conclude for \eqref{theorem:wh J*-J*-2-eq1}. Now, to show \eqref{theorem:wh J*-J*-2-eq2}, we observe from \eqref{theorem:wh J*-J*-2-eq1} that
\begin{eqnarray*}
\E\,|\wh J^*_{T_0} -J^*(T_0,x,Y_{T_0})| &\le&  x^\gamma\, \Gamma_1\, \left((2 \, \iota_0)^\gamma+ (\E|Y_{T_0}|)^\gamma \right) \,
\E\varpi^{\gamma}
\nonumber\\[2mm]
&+&x^\gamma \,\Gamma_2 \,(2 \, \iota_0)^\gamma \,\e(T_0)^\gamma+\E (|Y_{T_0}|^\gamma \,|\ov \alpha|^\gamma)\,.
\end{eqnarray*}
Taking into account that $\E\varpi\le \e_{2} \,(T_0)$ and using
 the bound \eqref{bound:E-Ygamma-alphagamma} we obtain \eqref{theorem:wh J*-J*-2-eq2}.
\end{proof}

\begin{lemma}\label{lemma:XT_bound-2}
The  estimate of the wealth process $(\wh X^{*}_t)_\zs{T_\zs{0}\le t\le T} $ satisfies the following inequalities
\begin{eqnarray} \nonumber
 \sup_{T_0\le t \le T}\E_\zs{T_0} &|\wh X^{*}_t - X^{*}_t|^\gamma \le
 x^\gamma \,k_3 \,(2 \, \iota_0+ |Y_{T_0}|)^\gamma \,\varpi^\gamma
\\[2mm] \label{eq:boundXgamma-2}&
+ \,x^\gamma \, k_4 \,(2 \, \iota_0+ |Y_{T_0}|)^\gamma \,|\ov \alpha|^\gamma
 \end{eqnarray}
 and
\begin{align}\nonumber
  \E_\zs{T_0} \,\left(\int^T_\zs{T_0}
| \wh c^*_\zs{t}\,\wh X^{*}_\zs{t}- c^*_\zs{t}\, X^{*}_\zs{t} |^{\gamma}\, \d t \right)
&\le x^\gamma\,k_5 \,(2 \, \iota_0+ |Y_{T_0}|)^\gamma \,\varpi^\gamma
\\[2mm]\label{lemma:XT_bound-2-eq-2}
&+x^\gamma \, k_6 \,(2 \, \iota_0+ |Y_{T_0}|)^\gamma \,|\ov \alpha|^\gamma \,.
 \end{align}
\end{lemma}

\begin{proof}
We follow the arguments in Lemma \ref{lemma:XT_bound} we set
$\Delta_t=\wh X^{*}_\zs{t}- X^{*}_\zs{t}$, $g(t)=\E_\zs{T_0}(\Delta_t^2)$ we get
\begin{eqnarray*}
 g(t) &\le& c_0 \int_{T_0}^t g(s) \, \d s + \psi(t),
\end{eqnarray*}
where $\psi(t)= 4 \, \E_\zs{T_0}\int_{T_0}^t \left( |\wh A_s - A_s |^2 +|\wh B_s - B_s |^2\right)\,|\wh X_s^*|^2 \, \d s$.
Through the Gronwall-Bellman inequality we get
\begin{eqnarray*}
 g(t) &\le& \psi(t) e^{c_0 t}
\le x^2\,\wt c \int_{T_0}^t \E_\zs{T_0} \left( |\wh A_s - A_s |^2 +|\wh B_s - B_s |^2\right)\,\d s\\
&\le& x^2\,\wt c \wt T \dfrac{2 (2 \mu_2+r)^2+\sigma_1^2}{\sigma_1^4 (1-\gamma)^2}\,\varpi^2
+2\,x^2\,\wt c \int_{T_0}^t \E_\zs{T_0}|\wh h(s,Y_s)^{-q_*}-h(s,Y_s)^{-q_*} |^2\, \d s\\
&\le& 2\,x^2\,\wt c \wt T\, \left( \dfrac{2\mu_2+ r+\sigma_1}{\sigma_1^2 (1-\gamma)}\right)^2\,\varpi^2
+ 2\,x^2\,\wt c q_* \int_{T_0}^t \E_\zs{T_0}|\wh h(s,Y_s)-h(s,Y_s)|^2 \, \d s \,.
\end{eqnarray*}
We use then Proposition \ref{prop:h-hate-converges-to-h-2} to get the analogous of \eqref{eq:bond_wh_h-h}:
\begin{eqnarray}\label{eq:bond_wh_h-h-2}
\E_\zs{T_0}|\wh h(s,Y_s)-h(s,Y_s)|
\le e^{\varkappa (T-s)}\,(2\,\iota_0+|Y_{T_0}| ) \, \wt \delta \,,
\end{eqnarray}
where $\wt \delta=\h_\zs{2} \,  \varpi +\h_\zs{1}\, |\ov \alpha|$.
Then
\begin{eqnarray*}
g(t) &\le& 2\,x^2\,\wt c \,\wt T\, \left( \dfrac{2 \mu_2+ r+\sigma_1}{\sigma_1^2 (1-\gamma)}\right)^2\, \varpi^2
+
2\,x^2\,\wt c\, q_*\,\frac{e^{2\,\varkappa \wt T}}{\varkappa^2}\,(2\,\iota_0+|Y_{T_0}| )^2 \,\wt \delta^2\\[2mm]
&\le& x^2\,\left(k'_1 \,\varpi +k'_2 \,(2\,\iota_0+|Y_{T_0}| )\, \wt \delta\, \right)^2 \,,
\end{eqnarray*}
i.e.
$$
\E_\zs{T_0}\vert \Delta_\zs{t}\vert\le
x\,\left(k'_1 \,\varpi +k'_2 \,(2\,\iota_0+|Y_{T_0}| )\, \wt \delta\, \right)\,.
$$
Using here the Jensen inequality for power function $z^{\gamma}$ (with $0<\gamma<1$) we obtain \eqref{eq:boundXgamma-2}. Now, we show \eqref{lemma:XT_bound-2-eq-2}. We follow the same arguments used in Lemma \ref{lemma:XT_bound} to arrive at
\begin{eqnarray*}
\E\,\left(\int^T_\zs{T_0}
| \wh c^*_\zs{t}\,\wh X^{*}_\zs{t}- c^*_\zs{t}\, X^{*}_\zs{t} |^{\gamma}\, \d t \right)
&\le&
x^\gamma \,q_*\,\wt d^\gamma \,  \int^T_\zs{T_0}
\E\, |\wh h(s,Y_s)-h(s,Y_s)|^\gamma \,\d s\\
&+&
\wt T \,\sup_{T_0 \le t \le T} \E|\wh X^{*}_\zs{t}- X^{*}_\zs{t} |^{\gamma}\,.
\end{eqnarray*}
Therefore, the upper bound \eqref{lemma:XT_bound-2-eq-2} follows immediately from
 \eqref{eq:boundXgamma-2} and \eqref{eq:bond_wh_h-h-2}.
\end{proof}

\section{Simulation}\label{sec:simulation}
In this section we use Scilab for simulations. In Fig 1. we simulate the truncated sequential estimate $\wh \alpha$ for different values of $T_0$, through 30 paths of the driving process $Y$. The sequential estimates are represented by $\times$ for $T_0=5$ days and $\ast$ for $T_0=10$ days. The true drift value of the process $Y$ is $\alpha=-5$. We take the bounds $\alpha \in [-0.15,-10]$ and set $\beta=1$.
\begin{center}

\includegraphics[height=60mm]{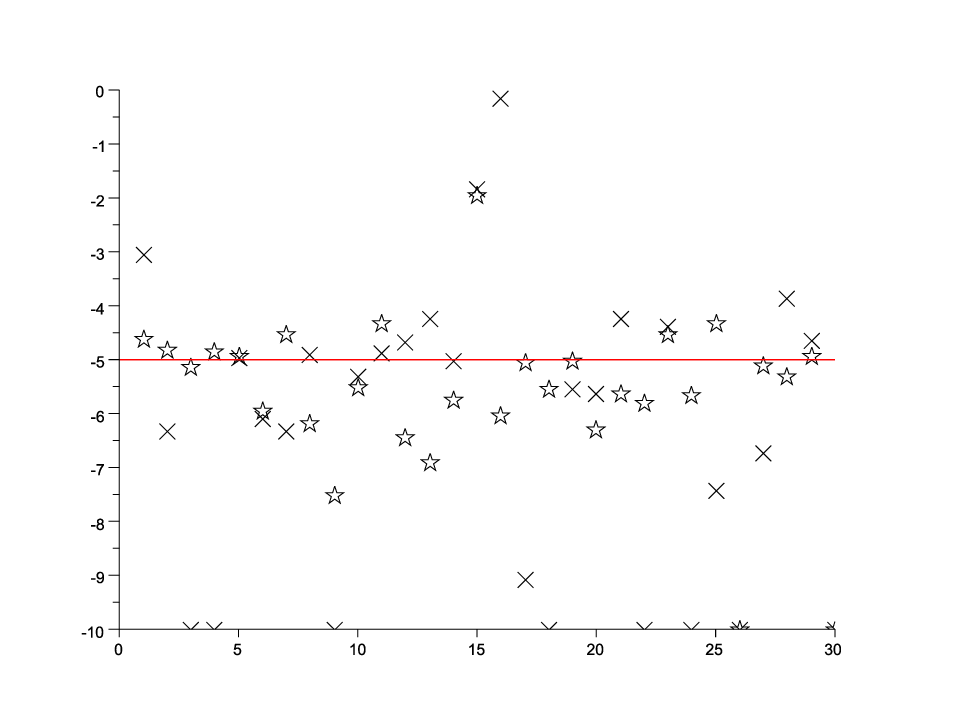}

\textbf{Fig 1: The truncated sequential estimate for $T_0=5$, $T_0=10$}
\end{center}

\begin{center}

\includegraphics[height=60mm]{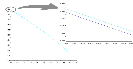}

  \textbf{Fig 2: The limit functions $h(t,0)$ and $\wh h(t,0)$}
\end{center}

In Fig 2. we simulate the limit functions $h(t,y)$ and $\wh h(t,y)$, under the following market settings: we set  $T_0=5$ and $\wt T=T-T_0=1$, $r=0.01$, $\mu=0.02$. The volatility is defined by
$\sigma(y)=0.5+\sin^2(\,y \,)$. The utility parameter is $\gamma=0.75$. To simulate $\wh h(t,y)$, we use a very pessimistic realization of the truncated estimate ie; $\wh \alpha=-0.5$. The true value is $\alpha=-5$.  We see that, even in this extreme situation, the estimated function $\wh h(t,y)$ does not deviate significantly from the real value $h(t,y)$.


\Appendix
\section{Auxiliary results}

\subsection{Bounds for $h$ and  $\cH_\zs{h}$}
Let $h$ the fixed point solution for $h=\cL_h$, where the mapping  $\cL$ is defined in
\eqref{def:L-theOperator} and
\eqref{def:Hf(t,s,y)}.
Now we study the partial derivative of the $\cH_\zs{f}(t,s,y)$ with respect to $y$.
\begin{lemma}\label{Le.sec:A.1}
For any  $T_0<t\le s\le T$
\begin{equation}
\label{sec:A.7}
\sup_\zs{y\in\bbr}
\sup_\zs{f\in \cX}
\left|
\frac{\partial }{\partial y}
\,\cH_\zs{f}(t,s,y)\,
\right|
\le Q_1^* \wt T e^{Q_* \wt T}+ \frac{e^{Q_* \wt T}}{ \nu_s}\,,
\end{equation}
where
$
\nu_s^2=\beta^2 (1-e^{2 \alpha (s-t)})/2 |\alpha|\,.
$
\end{lemma}

\begin{proof}
To calculate this conditional expectation note, first that
$$
 \eta_s = y e^{\alpha (s-t)} + \int_t^s \beta e^{\alpha (s-v)} \, \d \wt \U_v = y e^{\alpha (s-t)} + \xi_s\,.
$$
Since $\eta$ it is a \emph{gaussian process}, for any $t< v_\zs{1}<\ldots <v_\zs{k}< s$ and
for any bounded $\bbr^{k}\to\bbr$ function $G$
\begin{equation}\label{sec:A.11}
\E
\left(
G(\eta_\zs{v_\zs{1}},\ldots,\eta_\zs{v_\zs{k}})
|\eta_\zs{s}=z
\right)=
\E
\,
G(\B_\zs{v_\zs{1}},\ldots,\B_\zs{v_\zs{k}})
\,,
\end{equation}
where 
$
\B_\zs{v}=\eta_v-k(v)\,\eta_s+k(v)\,z
$
and the coefficient $k(v)$ is chosen so that
$$
\E \, \left( \xi_v- k(v) \,\xi_s \right)\, \xi_s =0
\quad\mbox{i.e.}\quad
 k(v)=\frac{\E \xi_v \xi_s}{\E \xi_s^2}=e^{\alpha(s-v)}\frac{1-e^{2\alpha(v-t)}}{1-e^{2\alpha(s-t)}}\le 1\,.
$$

\noindent The conditional expectation with respect to
$\eta_{s}$ lets represent $\cH_\zs{f}$ as
\begin{equation}\label{sec:A.10}
\cH_\zs{f}(t,s,y)
=
\int_\zs{\bbr}\,\wh{\cH}_\zs{f}(s,y,z)\, \p(z,y)
\, \d z\,,
\quad
\p(z,y)=\frac{1}{\nu_s \,\sqrt{2\pi}}\,
\exp \left(-\frac{(z-\mu(y))^{2}}{2 \,\nu_s^{2}} \right)\,,
\end{equation}
where $\mu(y)=\E \, \eta_s= y \, e^{\alpha(s-t)}$ and $\nu_s^{2}= Var\, \eta_s$. Since $\B_s=z$
we get
\begin{eqnarray}
 \wh{\cH}_\zs{f}(s,y,z)&=&\E \left(
\left(f(s,\eta^{t,y}_{s})\right)^{1-q_\zs{*}}
\, \exp {\left( \int_{t}^{s} Q(\eta_{u}^{t,y})\, \d u \right)}
|\eta_\zs{s}=z
\right)\nonumber\\
&=&
\E\left(
\left(f(s,z)\right)^{1-q_\zs{*}}
\, \exp {\left( \int_{t}^{s} Q(\B_{u})\, \d u\right)}
\right)
\le e^{Q_* (s-t)} \,.
\end{eqnarray}
From here it follows that
$$
\left| \frac{\partial }{\partial y} \wh{\cH}_\zs{f}(s,y,z)\right| \le \left| \int_t^s \frac{\partial Q( \B_u) }{\partial y} \, \d u \right|  \wh{\cH}_\zs{f}(s,y,z)
\le Q_1^* (s-t) \, e^{Q_* (s-t)} \le Q_1^* \, \wt T \,e^{Q_* \wt T} \,.
$$
Now from \eqref{sec:A.10} we obtain
$$
\frac{\partial \cH_\zs{f}(t,s,y)}{\partial y}=
\int_\zs{\bbr}\,\frac{\partial\wh{\cH}_\zs{f}(s,y,z)}{\partial y}\,
\p(z,y)
\, \d z
+\int_\zs{\bbr}\,\wh{\cH}_\zs{f}(s,y,z)\,
\frac{(z-\mu(y)) \,\mu'(y)}{ \nu_s^{2}}\p(z,y)\, \d z
\,.
$$
Therefore,
\begin{align*}
\left |\frac{\partial \cH_\zs{f}(t,s,y)}{\partial y} \right|
&\le Q_1^* (s-t) e^{Q_* (s-t)}+ e^{Q_* (s-t)} \frac{\mu'(y)}{ \nu_s^{2}}\int_\bbr |z-\mu(y)| \p(z,y) \d z\\[2mm]
&\le
 Q_1^* (s-t) e^{Q_* (s-t)}+ \frac{e^{(Q_* +\alpha)(s-t)}}{ \nu_s^{2}}\frac{2 \nu_s}{\sqrt{2 \pi }}
\le Q_1^* \wt T e^{Q_* \wt T}+ \frac{e^{Q_* \wt T}}{ \nu_s}\,.
\end{align*}
Hence Lemma \ref{Le.sec:A.1}.
\end{proof}

\vspace{2mm}
\begin{lemma} \label{lemme:bonde-of-derive-h}
For any $y \in \bbr$, the unique solution of the fixed point equation $h= \cL_\zs{h}$ is differentiable with respect to $y$, and its partial derivative is bounded:
\begin{equation}
\label{sec:A.nn-1}
\sup_{T_0\le t  \le T, y\in \bbr}\left|\, \frac{\partial }{\partial y }  h(t,y)\right|
\le h_1^*\,,
\end{equation}
where $h_1^*$ is given in \eqref{def:h*-1}.
\end{lemma}
\begin{proof} It is obviously sufficient to show that $\cL_h(t,y)$ is differentiable with respect to $y$, and its partial derivative is bounded:
$$
\sup_{T_0\le t  \le T, y\in \bbr}\left|\, \frac{\partial }{\partial y }  \cL_\zs{f}(t,y)\right|
\le h_1^*\,.
$$
From the definition of $\cL_f$ in \eqref{def:L-theOperator}, for all $f \in \cX$ and for all $t \in [T_0,T]$ and $y \in \bbr$ we get

$$
\frac{\partial }{\partial y }
\cL_\zs{f}(t,y)
= \E
\frac{\partial }{\partial y}
\cG(t,T,y)
+\frac{1}{q_\zs{*}}\,
\int_{t}^{T}\frac{\partial }{\partial y}
\cH_\zs{f}(t,s,y)\,\d s \,.
$$
Using lemmas \ref{Le.sec:A.1} and \ref{Le.sec:A.4}, we get
\begin{eqnarray*}
\sup_{T_0\le t  \le T, y\in \bbr}\left|\, \frac{\partial }{\partial y }  \cL_\zs{f}(t,y)\right|
&\le& \wt T\, Q_\zs{1}^*\,e^{Q_\zs{*} \wt T} +\frac{1}{q_\zs{*}}\,
\int_{t}^{T}Q_1^* \wt T e^{Q_* \wt T}\, \d s\,
+\frac{1}{q_\zs{*}}\,\int_{t}^{T} \frac{e^{Q_* \wt T}}{ \nu_s} \,\d s \\
&\le& \wt T\, Q_\zs{1}^*\,e^{Q_\zs{*} \wt T} +\frac{Q_1^*\wt  T^2}{q_\zs{*}}\, e^{Q_* \wt T}
+\frac{e^{Q_* \wt T}}{q_\zs{*}}\,\int_{t}^{T}\frac{1}{\nu_s}\,\d s \,.
\end{eqnarray*}
To estimate $\int_{t}^{T}(1/\nu_s)\,\d s$ we observe that  $2|\alpha|(s-t)\le 2|\alpha| \, \wt T $ so
$$
\nu_s^2=\beta^2 \dfrac{(1-e^{2 \alpha (s-t)})}{2 |\alpha|(s-t)} (s-t)  \ge \beta^2\dfrac{(1-e^{2 \alpha})}{2 |\alpha|}(s-t)
\quad\mbox{if}\quad (s-t) \le 1
$$
and
$$
\nu_s^2=\beta^2 \dfrac{(1-e^{2 \alpha (s-t)})}{2 |\alpha|} \ge \beta^2\dfrac{(1-e^{2 \alpha})}{2 |\alpha|}
\quad\mbox{if}\quad (s-t) \ge 1\,.
$$
Therefore, we get
\begin{eqnarray*}
\int_t^T \frac{1}{\nu_s} \, \d s &\le& \sqrt{\dfrac{2 |\alpha|}{\beta^2(1-e^{2 \alpha})}} \int_t^{t+1}\frac{1}{\sqrt{s-t}} \, \d s+\sqrt{\dfrac{2 |\alpha|}{\beta^2(1-e^{2 \alpha})}}\int_{t+1}^T \, \d s \\[5mm]
&\le&
2 \sqrt{\dfrac{2 |\alpha|}{\beta^2(1-e^{2 \alpha})}} +\sqrt{\dfrac{2 |\alpha|}{\beta^2(1-e^{2 \alpha})}} \wt T
\le 3 \sqrt{\dfrac{2 |\alpha|}{\beta^2(1-e^{2 \alpha})}} \wt T\,.
\end{eqnarray*}
Taking into account  that $\alpha_2 \le \alpha \le \alpha_1<0$ we obtain the desired result.
\end{proof}

\subsection{Properties of the function $\cG$}

Now we study the partial derivatives of the function $\cG(t,s,y)$ defined in
\eqref{def:L-theOperator}. To this end we need the following general result.

\begin{lemma}\label{Le.sec:A.3}
Let $F=F(y,\omega)$ be a $\bbr\times \Omega\to\bbr$ random bounded
function
such that for some nonrandom constant $c^{*}$
$$
\left|
\frac{\d }{\d y}\,F(y,\omega)
\right|\,\le\,
c^{*}
\quad\mbox{a.s.}
\,.
$$
Then
$$
\frac{\d }{\d y}\,\E\,F(y,\omega)=
\E\,\frac{\d }{\d y}\, F(y,\omega)
\,.
$$
\end{lemma}
\noindent This Lemma follows immediately from the Lebesgue dominated convergence theorem.

\begin{lemma}\label{Le.sec:A.4}
For any $t<s$ the function $\cG$ satisfies the following properties:
\begin{equation}\label{sec:A.18}
\sup_\zs{y\in\bbr}\,
\left|
\frac{\partial \cG(t,s,y)}{\partial y}
\right|\,\le (s-t) Q_\zs{1}^*\,
e^{Q_\zs{*}(s-t)}
\quad\mbox{and}\quad
\frac{\partial }{\partial y}\,\E\,\cG(t,s,y)=
\E\,\frac{\partial }{\partial y}\cG(t,s,y)
\,.
\end{equation}
\end{lemma}

\begin{proof}
We have immediately
$$
\frac{\partial \cG(t,s,y)}{\partial y}=\cG(t,s,y)\G(t,s,y) \,,
$$
where
$
\G(t,s,y)=
\int^{s}_\zs{t}Q_\zs{1}(\eta^{t,y}_\zs{u})\,(\partial \, \eta_\zs{u}^{t,y}/\partial \, y)
\, \d u$ and $Q_\zs{1}(z)=\d Q(z)/\d z$.
Now Lemma~\ref{Le.sec:A.3} imply directly this lemma.
\end{proof}
\begin{lemma}\label{lemma_4.19*}
 For any $f \in \cX$ having bounded partial derivatives with respect to $y\in \bbr$ and for any $N>0$ we have
$$
\sup_{|y|<N} \,\,\sup_{0 \le t_1 < t_2 \le T} \frac{|\cL_f (t_2,y)- \cL_f(t_1,y)|}{\sqrt{t_2 -t_1}} < \infty.
$$
\end{lemma}
\begin{proof} One can check directly that
$$
\sup_{|y| <N} \,\, \sup_{0 \le t_1 < t_2 \le s \le T} \frac{\E\,|\eta_s^{t_2,y}-\eta_s^{t_1,y}|}{\sqrt{t_2 -t_1}} < \infty.
$$
This upper bound implies directly Lemma \ref{lemma_4.19*}.
\end{proof}
\subsection{Properties of the process $\eta$}

We recall that to the process $(\eta_s)_{0\le s \le T}$ is defined in \eqref{def:EDS-eta} and $(\wh \eta_s)_{0\le s \le T}$ defined in
\eqref{def:EDS-wh-eta}.

\begin{lemma}\label{lemme:bound-EY-and-EG(t,y)}
For any $T_0 \le t\le s \le T$
\begin{equation}\label{lemme:bound-EY-and-EG(t,y)-eq1}
\E_\zs{T_0} |\wh \eta_s^{t,y}| \le
\iota_0+|y|=\dfrac{\beta}{\sqrt{2 |\alpha_1|}}+|y|:=
\wt m(y)
\end{equation}
and
\begin{equation}\label{lemme:bound-EY-and-EG(t,y)-eq2}
 \E_\zs{T_0} \int_{t}^T |\wh \eta_s^{t,y}- \eta_s^{t,y}| \, \d t \le\, \dfrac{ \wt   T \wt m (y)}{|\alpha_1|}
\,|\ov \alpha| \,.
\end{equation}
Moreover, for the  known parameter $\mu$ and unknown parameter $\alpha$
\begin{eqnarray}\label{lemme:bound-EY-and-EG(t,y)-eq3}
 \E_\zs{T_0}   |\wh \cG(t,s,y)-\cG(t,s,y)|&\le& \wt T Q_1^* e^{Q_* (T-t)}  \, \frac{ \wt m(y)  }{|\alpha_1|}\, | \ov \alpha| \,,
\end{eqnarray}
where $Q^*$ and  $Q_1^*$ are defined in \eqref{def:Q*-Q_1*}, and $\wh \cG(t,s,y)$ is given in \eqref{def:wh-L-theOperator}.
\end{lemma}
\begin{proof}
Since $\eta_s = \eta_t e^{\alpha \,(s-t)} + \int_t^s \beta e^{\alpha (s-v)} \, \d \wt \U_v$
 we obtain for any  $\alpha_2 \le \alpha \le \alpha_1<0$
\begin{eqnarray*}
 \E\,(\eta_s^{t,y})^2&=& y^2 e^{2 \alpha \,(s-t)}+\beta^2 \int_t^s e^{2 \alpha (t-v)} \, \d v \le y^2+\dfrac{\beta^2}{2 |\alpha_1|}\\
&\le& \left( |y|+\dfrac{\beta}{\sqrt{2 |\alpha_1|}}\right)^2 \,.
\end{eqnarray*}
This implies the bound \eqref{lemme:bound-EY-and-EG(t,y)-eq1}. Moreover, setting  $\ov \eta_s=\wh \eta_s^{t,y}- \eta_s^{t,y}$,
we obtain
$$
\d \ov \eta_s=(\wh \alpha \wh \eta_s^{t,y}-\alpha \eta_s^{t,y}) \, \d s
= \alpha \ov \eta_\zs{s}\, \d s +(\ov \alpha) \wh \eta_s^{t,y} \, \d s\,,
$$
i.e. $\ov \eta_s= \int_t^s \ov \alpha \, e^{\alpha \, (s-u)} \, \wh \eta_u^{t,y} \, \d u$. Therefore,
$$
|\ov \eta_s| \le |\ov \alpha| \int_t^s  |\wh \eta_u^{t,y}| e^{\alpha \, (s-u)} \, \d u\,.
$$
Since $\wh \alpha$ is independent of the Brownian motion $(\wt \U_t)$, we get
\begin{eqnarray}\label{eq:E_hate_eta-eta_inequality}
 \E_\zs{T_0}\,  |\ov \eta_s|
&\le&   |\ov \alpha|\, \E_\zs{T_0}  \int_t^s  |\wh \eta_u^{t,y}| \, e^{\alpha \, (s-u)} \, \d u \nonumber\\
&\le&  |\ov \alpha|  \int_t^s e^{\alpha \, (s-u)} \E_\zs{T_0}  |\wh \eta_u^{t,y}| \, \d u
\le   \,\frac{\wt m (y)}{|\alpha_1|}|\ov \alpha| \,.
\end{eqnarray}
Therefore, for all $ \, T_0\le t\le T$

\begin{eqnarray*}
\E_\zs{T_0}\int_t^T |\ov \eta_s| \, \d s
&\le& \E_\zs{T_0}\left( \int_t^T |\ov \alpha| \int_t^s  e^{\alpha (s-u)} |\wh \eta_u^{t,y}| \, \d u \, \d s\right)\\[2mm]
&\le& \wt T |\ov \alpha| \int_t^T e^{\alpha (s-u)}\,\E_\zs{T_0} |\wh \eta_u^{t,y}| \, \d u
\le \frac{\wt m (y) \wt T  }{|\alpha_1|}|\ov \alpha|
\end{eqnarray*}
and we come to \eqref{lemme:bound-EY-and-EG(t,y)-eq2}. To get inequality \eqref{lemme:bound-EY-and-EG(t,y)-eq3} note that
\begin{eqnarray*}
 |\wh \cG(t,s,y)-\cG(t,s,y)|
&\le& e^{Q_* (T-t)} \, |\int_{t}^{s} Q(\wh \eta_{u}^{t,y})\, \d u -\int_{t}^{s} Q(\eta_{u}^{t,y})\, \d u| \\
&\le& e^{Q_* (T-t)}  \, \int_{t}^{s} \sup_{y \in \bbr} |\frac{ \partial Q(y)}{\partial y}|\, |\ov \eta_{u}|\, \d u \\
&\le& Q_1^* e^{Q_* (T-t)} \int_{t}^T|\ov \eta_{u}|\, \d u\,.
\end{eqnarray*}
Thus,
$$
\E_\zs{T_0}  |\wh \cG(t,s,y)-\cG(t,s,y)|
\le
 Q_1^* e^{Q_* (T-t)} \int_{t}^T \E_\zs{T_0} |\ov \eta_{u}|\, \d u\,.
$$
Now, the bound \eqref{eq:E_hate_eta-eta_inequality} implies \eqref{lemme:bound-EY-and-EG(t,y)-eq3}.
Hence Lemma \ref{lemme:bound-EY-and-EG(t,y)}.
\end{proof}

\vspace{3mm}
\noindent We study in the next proposition the behavior of $h(t,y)$, the solution of the fixed point problem $h=\cL_h$, when using the estimate $\wh\alpha$ of the parameter $\alpha$. We look for a bound for the deviation
\begin{equation}\label{sec:A.nn-2}
\ov h(t,y)= \wh h(t,y)-h(t,y)\,,
\end{equation}
 where $\wh h=\wh{\cL}_{\wh h}$. The operator $\wh \cL$ is defined in \eqref{def:wh-L-theOperator}.
Similarly to \eqref{def:norme_in_cX} we define on $\cX$ the metric $ \wt \varrho_\zs{*}$ as follows:
\begin{equation}\label{def:metric_wt_varrho*}
\wt \varrho_\zs{*}(f,g)=\sup_\zs{(t,y)\in \cK} \,e^{-\varkappa(T-t)}\,\frac{|f(t,y)-g(t,y)|}{\iota_0+|y|} \,,
\end{equation}
where we set $\iota_0=\beta/\sqrt{2 \, \alpha_1}$ and $\varkappa=Q_*+\zeta+1$ and set $\zeta=\zeta_0+2\gamma$ for some $\zeta_0 >0$.
\vspace{2mm}
\begin{proposition}\label{prop:h-hate-converges-to-h}
For the  known $\mu$ and unknown $\alpha$, and for any  $0<T_0 <T$
\begin{equation}\label{eq:varrho_h-h}
\wt \varrho_*(\wh h,h)\le \h_{1}\,\,|\ov \alpha|\,,
\end{equation}
where $\h_{1}$ is given in
\eqref{def:h_T-Gamma_T}.
\end{proposition}

\begin{proof} We use the definition of the operator $\cL$ in \eqref{def:L-theOperator}:
$$
h(t,y)=\cL_h(t,y)=\E\,\cG(t,T,y)
+\frac{1}{q_\zs{*}}
\int_{t}^{T}
\cH_\zs{h}(t,s,y)
\,\d s\,.
$$
Through  \eqref{def:Hf(t,s,y)} we can estimate the deviation
 \eqref{sec:A.nn-2} as
$$
 |\ov h(t,y)|
 \le
 \E_\zs{T_0} \,| \wh \cG(t,T,y)-\cG(t,T,y)| +I(\wh \alpha)\,,
$$
where
$$
I(\wh \alpha )=
\frac{1}{q_\zs{*}}
\int_{t}^{T} \E_\zs{T_0} \,
|\left(\wh h(s,\wh \eta^{t,y}_{s})\right)^{1-q_\zs{*}}
\wh \cG(t,s,y)\,-\left(h(s,\eta^{t,y}_{s})\right)^{1-q_\zs{*}}
\cG(t,s,y)|
\,\d s\,.
$$
Moreover, this term can be bounded as
\begin{eqnarray*}
I(\wh \alpha )&\le& \frac{1}{q_\zs{*}}
\int_{t}^{T} \E_\zs{T_0} \, (h(s,\eta^{t,y}_{s})^{1-q_\zs{*}}\,
|\wh \cG(t,s,y)-\cG(t,s,y)|\,\d s\,\\
&+& \frac{1}{q_\zs{*}}
\int_{t}^{T} \E_\zs{T_0} \,
|\left(\wh h(s,\wh \eta^{t,y}_{s})\right)^{1-q_\zs{*}}
-\left(h(s,\eta^{t,y}_{s})\right)^{1-q_\zs{*}}
| e^{Q_* (s-t)}
\,\d s\,\\
&\le&
\int_{t}^{T} \E_\zs{T_0} \,
|\wh \cG(t,s,y)-\cG(t,s,y)|\,\d s\,
+\frac{|1-q_*|}{q_\zs{*}}
\int_{t}^{T} \E_\zs{T_0} \,|\wh h(s,\wh \eta^{t,y}_{s})
-h(s,\eta^{t,y}_{s})| e^{Q_* (s-t)}\,\d s\,.
\end{eqnarray*}
We use the fact that $q_*=1/(1-\gamma)>1$ and the bounds \eqref{lemme:bound-EY-and-EG(t,y)-eq3} and \eqref{eq:E_hate_eta-eta_inequality} to deduce
\begin{eqnarray*}
|\ov h(t,y)|&\le&(1+\wt T)\E_\zs{T_0} \,| \wh \cG(t,T,y)-\cG(t,T,y)|
+\gamma \,\int_{t}^{T} \E_\zs{T_0} \,|h(s,\wh\eta^{t,y}_{s})
-h(s, \eta^{t,y}_{s})| e^{Q_* (s-t)}\,\d s\\
&+&\gamma \,
\int_{t}^{T} \E_\zs{T_0} \,|\wh h(s,\wh \eta^{t,y}_{s})
-h(s,\wh \eta^{t,y}_{s})| e^{Q_* (s-t)}\,\d s\,.
\end{eqnarray*}
The upper bound \eqref{sec:A.nn-1} yields
$$
|h(s,\wh\eta^{t,y}_{s})-h(s, \eta^{t,y}_{s})| \le \, h_1^* \,\,  |\wh\eta^{t,y}_{s}- \eta^{t,y}_{s}|\,.
$$
In view of the definitions of the metric $\wt \varrho_*$ in \eqref{def:metric_wt_varrho*} and of the parameter $\varkappa$ in \eqref{def:norme_in_cX}
 we get
\begin{eqnarray*}
\wt \varrho_*(\wh h,h)
&\le&\frac{(1+\wt T)\wt T\, Q^*_1}{|\alpha_1|}\,\sup_\zs{(t,y)\in \cK}  \left( \frac{\wt m(y)}{\iota_0+|y|}e^{(Q_*-\varkappa)(T-t)}\right) |\wh\alpha-\alpha| \\
&+& \gamma \,\sup_\zs{(t,y)\in \cK}\int_{t}^{T} \, h_1^* \,\E_\zs{T_0} |\wh \eta^{t,y}_s-\eta^{t,y}_s|\,e^{(Q_*-\varkappa) (T-t)}\,\d s \\
&+&\gamma \,\sup_\zs{(t,y)\in \cK}\int_{t}^{T} \E_\zs{T_0} \,\frac{|\wh h(s,\wh \eta^{t,y}_{s})-h(s,\wh \eta^{t,y}_{s})| e^{-\varkappa(T-s)} }{\iota_0+|\wh \eta^{t,y}_{s} |} \frac{\iota_0+|\wh \eta^{t,y}_{s} |}{\iota_0+|y|}e^{(Q_*-\varkappa) (s-t)}\,\d s\,.
\end{eqnarray*}
Then
\begin{eqnarray*}
\wt \varrho_*(\wh h,h)
&\le&
\Upsilon^{*}\, \,|\wh\alpha-\alpha| +\gamma  \,\wt\varrho_*(\wh h,h)\,
\sup_\zs{(t,y)\in \cK} \int_{t}^{T}  \,\frac{\iota_0+\E_\zs{T_0}|\wh \eta^{t,y}_{s} |}{\iota_0+|y|}e^{(Q_*-\varkappa) (s-t)}\,\d s\\
&\le&
\Upsilon^{*}\,\, |\wh\alpha-\alpha| +\gamma \,\wt \varrho_*(\wh h,h)\,
\sup_\zs{(t,y)\in \cK} \left( \frac{\iota_0+\wt m(y)}{\iota_0+|y|}\int_{t}^{T}   \,e^{(Q_*-\varkappa) (s-t)}\,\d s \right)\\
&\le&
\Upsilon^{*}\,\, |\wh\alpha-\alpha| +\,
\frac{2\,\gamma }{\varkappa-Q_*}\, \wt \varrho_*(\wh h,h)\,.
\end{eqnarray*}
Here $\Upsilon^{*}= \left(2\, Q_1^*\,\wt T  + \gamma\, h_1^* \, \right)\,\wt T /|\alpha_1|$.
Hence we get
$$
\wt \varrho_*(\wh h,h)\le \frac{\varkappa-Q_*}{\varkappa-Q_*-2\, \gamma}\Upsilon^{*}\,\vert \ov \alpha\vert\,.
$$
Taking into account  that $\varkappa=Q_*+\zeta_0+2\gamma+1$,
 we obtain \eqref{eq:varrho_h-h}.  Hence Proposition \ref{prop:h-hate-converges-to-h}.
\end{proof}

\vspace{5mm}
\noindent We consider both the stock price appreciation rate $\mu \in [\mu_1,\mu_2]$, and the drift $\alpha$ of the economic factor $Y$ to be unknown. The next lemma gives the analogous of equation \eqref{lemme:bound-EY-and-EG(t,y)-eq3}.

\begin{lemma}\label{lemme:bound-EY-and-EG(t,y)-2}
For the unknown parameters  $\mu$ and $\alpha$ and for any  $0<T_0 <T$
\begin{equation}\label{eq:bound-EY-and-EG(t,y)-2}
 \E_\zs{T_0}|\wh \cG(t,s,y)-\cG(t,s,y)|\le \gamma \dfrac{(\mu_2+r+1) }{(1-\gamma) \sigma_1^2}\, \wt T e^{ Q_* (T-t)}  \, \varpi
  + \wt T Q_1^* e^{Q_* (T-t)}  \, \frac{ \wt m(y)  }{|\alpha_1|}\,|\ov \alpha|\,,
\end{equation}
where $\varpi=|\wh \mu-\mu|+|\wh \mu-\mu|^2$.
\end{lemma}
\begin{proof}
First note that for the function $Q$ defined in \eqref{def:q*_Q} we can obtain the following bound
$$
 \wh Q(z)-Q(z)=\frac{\gamma \left(\wh \theta^2(z)-\theta^2(z) \right)}{2(1-\gamma)}
\le
\frac{\gamma (\mu_2+r+1) }{(1-\gamma) \sigma_1^2}\,\varpi\,.
$$
So,
 \begin{align*}
 |\wh \cG(t,s,y)-\cG(t,s,y)|
 &\le \left|\exp \left( \int_{t}^{s} \wh Q(\wh \eta_{u}^{t,y})\, \d u\right)-\exp \left( \int_{t}^{s} Q(\wh \eta_{u}^{t,y})\, \d u\right) \right|\\[2mm]
&+
e^{Q_* (T-t)} \, \left|\int_{t}^{s} Q(\wh \eta_{u}^{t,y})\, \d u -\int_{t}^{s} Q(\eta_{u}^{t,y})\, \d u \right| \\[2mm]
&\le \wt T e^{ Q_* (T-t)} \dfrac{\gamma (\mu_2+r+1) }{(1-\gamma) \sigma_1^2}\,\varpi + Q_1^* e^{Q_* (T-t)} \int_t^T \,|\wh \eta_{u}^{t,y}-\eta_{u}^{t,y}|\, \d u\,.
\end{align*}
Through \eqref{lemme:bound-EY-and-EG(t,y)-eq2} we obtain \eqref{eq:bound-EY-and-EG(t,y)-2}. Hence, Lemma \ref{eq:bound-EY-and-EG(t,y)-2}.
\end{proof}

\vspace{2mm}
The next proposition is the analogous of Proposition \ref{prop:h-hate-converges-to-h}. The difference is that, in the proposition bellow, both $\mu$ and $\alpha$ are unknown.
\begin{proposition}\label{prop:h-hate-converges-to-h-2}
 For the unknown parameters $\mu$ and $\alpha$and for any  $0< T_0 <T$
$$
\wt \varrho_*( \wh h,h) \le \h_{1}\, |\ov \alpha|+ \h_2 \, \varpi\,,
$$
where the metric $\wt \varrho_*$ is given in \eqref{def:metric_wt_varrho*},  $\h_{1}$ and  $\h_{2}$  are
given in \eqref{def:h_T-Gamma_T}
and in \eqref{def:h2} respectively.
\end{proposition}

\begin{proof}
We follow the same arguments as in the proof of Proposition \ref{prop:h-hate-converges-to-h}
and use Lemma \ref{lemme:bound-EY-and-EG(t,y)-2} for the bound of $\E_\zs{T_0}\,| \wh \cG(t,T,y)-\cG(t,T,y)|$.
\end{proof}

\vspace{20mm}


\bibliographystyle{plain}
\bibliography{Biblio_BePe-seq}

\begin{thebibliography}{10}

\bibitem{BerdjanePergamen2012}
B.~Berdjane and S.~Pergamenshchikov.
\newblock Optimal consumption and investment for markets with random
  coefficients.
\newblock {\em Finance and stochastics}, 17(2):419--446, 2013.

\bibitem{CastanedaLeyvaHernandezHernandez2005}
N.~Castaneda-Leyva and D.~Hern\'{a}ndez-Hern\'{a}ndez.
\newblock Optimal consumption investment problems in incomplete markets with
  stochastic coefficients.
\newblock {\em SIAM, J. Control and Opt}, 44:1322--1344, 2005.

\bibitem{Delarue2002}
F.~Delarue.
\newblock On the existence and uniqueness of solutions to \textsc{FBSDE}s in a
  non-degenerate case.
\newblock {\em Stochastic Processes and their Applications}, 99(2):209 -- 286,
  2002.

\bibitem{DelongKluppelberg2008}
L.~Delong and C.Kl\"{u}ppelberg.
\newblock Optimal investment and consumption in a black-scholes market with
  l\'evy-driven stochastic coefficients.
\newblock {\em Annals of Applied Probability}, 18(3):879--908, 2008.

\bibitem{FlemingHernandezHernandez2003}
W.~Fleming and D.~Hern\'{a}ndez-Hern\'{a}ndez.
\newblock An optimal consumption model with stochastic volatility.
\newblock {\em Finance and Stochastics}, 7:245--262, 2003.

\bibitem{FlemingRishel1975}
W.~Fleming and R.~Rishel.
\newblock {\em Deterministic and stochastic optimal control}.
\newblock Applications of Mathematics, 1, Springer-Verlag, Berlin New York,
  1975.

\bibitem{FouquePapanicoloauSircar2000}
J.P. Fouque, G.~Papanicoloau, and R.~Sircar.
\newblock {\em Derivatives in Financial Markets with Stochastic Volatility}.
\newblock Cambridge University Press, Cambridge, 2000.

\bibitem{Freidlin1985}
M.I Freidlin.
\newblock {\em Functional Integration and Partial Differential Equations}.
\newblock Annals of Mathematics Studies, 1, Springer-Verlag, 1985.

\bibitem{HernandezHernandezShied2006}
D.~Hern\'{a}ndez-Hern\'{a}ndez and A.~Schied.
\newblock Robust utility maximization in a stochastic factor model.
\newblock {\em Statistics and Decisions}, 24(1):109--125, 2006.

\bibitem{JackwerthRubinstein1996}
J.~Jackwerth and M.~Rubinstein.
\newblock Recovring probability distributions from contemporaneous sequirity
  prices.
\newblock {\em J. Finance}, 40:455--480, 1996.

\bibitem{KabanovPergamenshchikov2003}
Yu.~M. Kabanov and S.~M. Pergamenshchikov.
\newblock {\em Two-Scale Stochastic Systems. Asymptotic Analysis and Control}.
\newblock Applications of mathematics. Stochastic modelling and applied
  probability, Springer-Verlag, Berlin Heidelberg New York, 2003.

\bibitem{KaratzasShreve1991}
I.~Karatzas and S.~E. Shreve.
\newblock {\em Brownian Motion and Stochastic Calcul}.
\newblock Springer, New York, 1991.

\bibitem{KaratzasShreve1998}
I.~Karatzas and S.E. Shreve.
\newblock {\em Methods of Mathematical finance}.
\newblock Springer, Berlin, 1998.

\bibitem{KluppelbergPergamenchtchikov2009}
C.~Kl\"{u}ppelberg and S.~M. Pergamenchtchikov.
\newblock Optimal consumption and investment with bounded downside risk for
  power utility functions.
\newblock In F.~Delbaen, M.~R\'asonyi, and C.~Stricker, editors, {\em
  Optimality and Risk: Modern Trends in Mathematical Finance. The Kabanov
  Festschrift}, pages 133--169. Springer, Heidelberg-Dordrecht-London-New York,
  2009.

\bibitem{PergamenKonev92}
V.~Konev and S.~Pergamenshchikov.
\newblock Estimation of the parameters of diffusion processes.
\newblock {\em Methods of Economical Analysis}, pages 3--31, 1992.

\bibitem{Korn1997}
R.~Korn.
\newblock {\em Optimal portfolios}.
\newblock World Scientific, Singapore, 1997.

\bibitem{KraftSteffensen2006}
H.~Kraft and M.~Steffensen.
\newblock Portfolio problems stopping at first hitting time with application to
  default risk.
\newblock {\em Math. Meth. Oper. Res}, 63:123--150, 2006.

\bibitem{LadyzenskajaSolonnikovUralceva1967}
O.A. Lady$\check{z}$enskaja, V.A. Solonnikov, and N.N. Ural'ceva.
\newblock {\em Linear and quasilinear equations of parabolic type (Translated
  from the Russian)}.
\newblock Translations of Mathematical Monographs, Vol. 23 American
  Mathematical Society, Providence, R.I, 1988.

\bibitem{LiptserShiryaev2000-I}
R.S. Liptser and A.N. Shiryaev.
\newblock {\em Statistics of Random Process I. General Theory}.
\newblock Springer-Verlag Berlin and Heidelberg Gmb H \& Co, Berlin, 2nd
  revised edition, 2000.

\bibitem{LiptserShiryaev2000-II}
R.S. Liptser and A.N. Shiryaev.
\newblock {\em Statistics of Random Process II. Applications}.
\newblock Springer-Verlag Berlin and Heidelberg Gmb H \& Co, Berlin, 2nd
  revised edition, 2000.

\bibitem{BardiMartinoCesaroni2010}
L.~Manca M.~Bardi, A.~Cesaroni.
\newblock Convergence by viscosity methods in multiscale financial models with
  stochastic volatility.
\newblock {\em SIAM J. Financial Math}, 1:230--265, 2010.

\bibitem{MaProtterYong1994}
J.~Ma, P.~Protter, and J.~Yong.
\newblock Solving forward-backward stochastic differential equations explicitly
  — a four step scheme.
\newblock {\em Probability Theory and Related Fields}, 98(3):339--359, 1994.

\bibitem{Merton1971}
R.~Merton.
\newblock Optimal consumption and portfolio rules in a continuous time model.
\newblock {\em Journal of Economic Theory}, 3:373--413, 1971.

\bibitem{Novikov1971}
A.~A Novikov.
\newblock Sequential estimation of the parameters of the diffusion processes.
\newblock {\em Theory Probability and Appl}, pages 394--396, 1971.

\bibitem{Pazy1983}
A.~Pazy.
\newblock {\em Functional Integration and Partial Differential Equations}.
\newblock Applied Mathematical Sciences, 44, Springer New York, 1983.

\bibitem{Pham2002}
H.~Pham.
\newblock Smooth solutions to optimal investment models with stochastic
  volatilities and portofolio constraints.
\newblock {\em Appl. Math. Optim}, 46:55--78, 2002.

\bibitem{PresmanSethi1991}
E.~L. Presman and S.~P. Sethi.
\newblock Risk-aversion behavior in consumption/investment problems.
\newblock {\em Math. Finance}, 1(1):101--124, 1991.

\bibitem{Rogers2013}
L.C.G Rogers.
\newblock {\em Optimal Investment}.
\newblock SpringerBriefs in Quantitative Finance, Springer-Verlag, 2013.

\bibitem{Rubinstein1985}
M.~Rubinstein.
\newblock Nonparametric tests of alternative option pricing models.
\newblock {\em J. Finance}, 51:1611--1631, 1985.

\bibitem{SethiTaksarPresman1992}
M.~I.~Taksar S.~P.~Sethi and E.~L. Presman.
\newblock Explicit solution of a general consumption/portfolio problem with
  subsistence consumption and bankruptcy.
\newblock {\em Journal of Economic Dynamic and Control}, 16:747--768, 1992.

\bibitem{SircarZariphopoulou2005}
R.~Sircar and T.~Zariphopoulou.
\newblock Bounds and asymptotic approximations for utility prices when
  volatility is random.
\newblock {\em SIAM Journal on Control and Optimization (SICON)},
  43(4):1328--1353, 2005.

\bibitem{Zariphopoulou2001}
T.~Zariphopoulou.
\newblock A solution approach to valuation with unhedgeable risk.
\newblock {\em Finance and stochastics}, 5:61--82, 2001.

\end{thebibliography}
\nocite{KaratzasShreve1991}
\nocite{SircarZariphopoulou2005}
\nocite{BardiMartinoCesaroni2010}
\end{document}